\documentclass[twocolumn]{autart} 
\usepackage{natbib}
\usepackage{cite}
\usepackage{url}
\usepackage{amsmath}
\usepackage{amssymb}
\usepackage{mathrsfs}
\usepackage{mathtools}
\usepackage{bm}
\usepackage{algorithm}

\usepackage{algorithmic}
\usepackage{graphicx}
\usepackage{xcolor}
\usepackage{booktabs}
\usepackage{bbm}
\usepackage{physics}
\usepackage{subfig}
\usepackage{wrapfig}
\usepackage{har2nat}
\usepackage{bbm}
\usepackage{epstopdf}
\usepackage{siunitx}
\usepackage{bm}
\usepackage{enumerate}

\newcommand{\mathletter}[1]{%
    \expandafter\newcommand\csname b#1\endcsname{\mathbb{#1}}
    \expandafter\newcommand\csname c#1\endcsname{\mathcal{#1}}
	\expandafter\newcommand\csname f#1\endcsname{\mathfrak{#1}}
	\expandafter\newcommand\csname t#1\endcsname{\widetilde{#1}}
	\expandafter\newcommand\csname h#1\endcsname{\widehat{#1}}
	\expandafter\newcommand\csname s#1\endcsname{\boldsymbol{#1}}
    \expandafter\newcommand\csname o#1\endcsname{\overline{#1}}
}%
\def\mathletters#1{\mathlettersB#1,,}
\def\mathlettersB#1,{\ifx,#1,\else\mathletter#1\expandafter\mathlettersB\fi}
\mathletters{A,B,C,D,E,F,G,H,I,J,K,L,M,N,O,P,Q,R,S,T,U,V,W,X,Y,Z}

\def\us{\underline{\sigma}}

\def\mt{\mathrm{tr}}

\def\bee{\begin{equation}}
	\def\ene{\end{equation}}
\def\been{\begin{equation*}}
	\def\enen{\end{equation*}}

\newtheorem{theo}{Theorem}
\newtheorem{lemma}{Lemma}

\newtheorem{remark}{Remark}

\newtheorem{assumption}{Assumption}

\graphicspath{{./figures/}}

\begin{document}
	\begin{frontmatter} 
		\title{Direct Data-Driven Linear Quadratic Tracking via Policy Optimization} 
		\author{Shubo~Kang}, 
		\ead{ksb20@mails.tsinghua.edu.cn}
		\author{Keyou~You\corauthref{cor}}
		\corauth[cor]{Corresponding author} 
		\ead{youky@tsinghua.edu.cn}
		\address{Department of Automation, and Beijing National Research Center for Information Science and Technology, Tsinghua University, Beijing 100084, China.}
		\begin{keyword}  
			Linear systems, direct data-driven approach, linear quadratic tracking, policy optimization, online control. 
		\end{keyword} 
	
	\begin{abstract} 
		Direct data-driven optimal control provides an elegant end-to-end paradigm, yet its real-time applicability is often hindered by the growing dimensionality of online decision variables. Recent breakthroughs, notably Data‑EnablEd Policy Optimization (DeePO), overcome this bottleneck for the Linear Quadratic Regulator (LQR) through sample‑covariance parameterization; however, extending this paradigm to Linear Quadratic Tracking (LQT) poses a fundamental challenge. The core difficulty stems from the intricate coupling between time‑varying references and the feedback‑feedforward policy structure, which prevents a direct application of constant‑dimension parameterization. We first introduce a reference‑decoupled reformulation of LQT that naturally accommodates the covariance parameterization, guaranteeing a fixed dimension of decision variables independent of data horizon. This formulation is proven to be exactly equivalent to the indirect certainty‑equivalence LQT solution. Leveraging this characterization, we develop offline and online DeePO algorithms. Theoretically, we prove global linear convergence for the offline algorithm using local gradient dominance and smoothness, and show that in the online setting the optimality gap decays linearly up to a bias term that scales inversely with the signal‑to‑noise ratio (SNR). Numerical simulations varify the theoretical results and illustrate the superior tracking performance of the proposed method.
	\end{abstract}

	\end{frontmatter}

	\section{Introduction}

	Optimal control problems, particularly the Linear Quadratic Regulator (LQR) and Linear Quadratic Tracking (LQT), serve as fundamental benchmarks in data-driven control. Existing data-driven optimal control methods are fundamentally classified into two distinct paradigms: indirect methods and direct methods. The former follows a classical ``identify-then-control" principle, whereas the latter seeks an end-to-end mapping from raw data to optimal control policies. Furthermore, depending on the data collection mechanism and the interaction mode with the environment, these algorithms can be broadly categorized into episodic and (fully) adaptive settings. While the episodic setting allows for pre-collected batch data or repeated trials with system resets, the adaptive setting presents a far more stringent challenge, requiring the controller to continuously refine its policy in real-time along a single, uninterrupted state trajectory while maintaining system stability.

	As a highly reliable and modular paradigm, indirect data-driven methods leverage mature theoretical frameworks of identification and model-based control synthesis to provide rigorous performance guarantees. A rich line of work has established non-asymptotic performance bounds, both in episodic settings \cite{abbasi2011regret, dean2020sample, mania2019certainty} and, more closely related to our setting, in adaptive schemes where the policy is updated at every time step \cite{wang2021exact, lu2023almost}.  By recursively refining the model estimate and re-solving the certainty-equivalence optimal control problem, these online indirect methods can achieve strong regret and tracking performance. Their success, however, relies on an explicit intermediate identification step: at each time instant, one must update the model parameters and then solve the associated algebraic Riccati equations to obtain the controller gains. This model-solve cycle, while theoretically well understood, naturally motivates an alternative end-to-end paradigm,  where the policy parameters are updated directly from streaming data without the need for explicit model reconstruction or Riccati solutions. Such a direct approach could potentially simplify the algorithmic structure, and reduce per-step computation.

	Direct data-driven control offers exactly such an end-to-end philosophy. Rooted in behavioral theory and Willems' Fundamental Lemma \cite{willems2005note}, early works parameterize the LQR problem directly from raw data, enabling optimal control without identifying system matrices \cite{de2019formulas, de2021low, dorfler2023certainty}. The data informativity framework is also used for characterizing the minimal data requirements for direct optimal control \cite{van2020data, van2020noisy}, with recent extensions to LQT \cite{trentelman2021informativity, zhou2025dynamic}. Despite their conceptual appeal, these formulations lack a iterative structure and are inherently episodic. Besides, the number of decision variables grows linearly with the data length, making them computationally intractable for online adaptive control. Another line of research, known as Policy optimization (PO) methods, naturally support iterative refinement. However, existing data-driven PO largely relies on zeroth-order gradient estimation, which requires numerous system rollouts and is thus ill-suited to continuous online learning along a single trajectory \cite{malik2020derivative}. 
	
	A significant step toward overcoming these barriers was recently achieved by proposing the Data-EnablEd Policy Optimization (DeePO) method \cite{zhao2025data,zhao2025policy}. By introducing a sample covariance parameterization, DeePO ensures that the dimension of the policy parameters remains constant irrespective of the data length, while the exact policy gradient can be estimated directly from a single closed-loop trajectory. This innovation successfully realizes the long-standing goal of online, direct, data-driven optimal control for the LQR problem. However, extending DeePO to the LQT problem is non-trivial: unlike LQR, the LQT policy involves both a feedback gain and a feedforward gain, and a time-varying reference trajectory introduces intricate coupling between the reference signal and the policy parameters. The resulting optimization landscape precludes a straightforward generalization and requires a careful re-engineering of the parameterization and problem formulation. This paper addresses exactly this gap. We show that the coupling can be elegantly resolved through a reference-decoupled reformulation, which then enables a rigorous extension of the covariance parameterization and the DeePO framework to the LQT setting.

	To tackle the coupling challenge, this paper first proposes a reference-decoupled reformulation of the LQT problem. The key insight is that the optimal feedback gain is independent of the reference trajectory. By considering multiple orthogonal reference points that span the state space and summing their costs, we obtain an aggregated optimization problem whose unique optimal solution exactly recovers both the optimal feedback and feedforward gains of the original LQT. This yields a compact formulation that depends solely on the policy parameters and system matrices, free from the future reference trajectory. Building on this, we extend the covariance parameterization to the LQT case, expressing the closed-loop system via sample covariances of the collected data. The resulting data-driven formulation has decision variables of fixed dimension, independent of the data length, thereby eliminating the dimension explosion problem of classical methods based on the Fundamental Lemma. Importantly, we prove that this data-driven problem is strictly equivalent to the standard indirect certainty-equivalence LQT, inheriting the same theoretical guarantees while circumventing explicit model identification and online solutions of algebraic Riccati equations.

	We then develop offline and online DeePO algorithms. The offline algorithm performs projected gradient descent using a closed-form policy gradient. By establishing equivalence to a model-based policy gradient iteration~\cite{zhao2023data} with a data-dependent scaling matrix and proving local gradient dominance and smoothness properties, we show global linear convergence. For the online adaptive setting, we design an algorithm that interleaves control, data augmentation, and policy updates along a single trajectory. A dynamic re-parameterization strategy ensures consistent gradient updates despite time-varying data matrices, while injected exploration noise maintains persistent excitation. We prove that, under bounded noise and persistent excitation, the optimality gap is bounded by a linearly decaying term plus a constant bias inversely proportional to the signal-to-noise ratio, and that the closed-loop state remains uniformly bounded. While preliminary work~\cite{kang2025linear} focused on offline set-point tracking, this paper provides a comprehensive online adaptive framework for time-varying LQT with full convergence analysis. Numerical simulations validate the theoretical findings and tracking performance.

	The remainder of this paper is organized as follows. Section~\ref{sec:preliminary} introduces the problem formulation. Section~\ref{sec:covariance_lqt} presents the reference-decoupled formulation and covariance parameterization. Sections~\ref{sec:offline_deepo} and~\ref{sec:online_deepo} detail the offline and online algorithms with convergence analyses. Section~\ref{sec:experiment} provides numerical experiments, and Section~\ref{sec:conclusion} concludes.

	\section{Problem Formulation} \label{sec:preliminary}

	In this section, we introduce some notations and preliminaries on the linear quadratic tracking (LQT) problem, as well as the policy optimization algorithm. Then, we formulate the problem in this paper.

	\subsection{The model-based LQT problem}
	Consider the following discrete-time linear time-invariant (LTI) system:
	\begin{equation} \label{eq:system_dynamics}
		x_{t+1} = A x_t + B u_t + w_t,
	\end{equation}
	where $x_t \in \bR^n$ is the system state, $u_t \in \bR^m$ is the control input, and $w_t \in \bR^n$ is the process noise at time $t$.

	In this work, we consider the linear quadratic tracking control on the system (A,B), the objective of which is to design a control policy $u_t = \pi(x_t, z_t,x_{t-1},z_{t-1},\dots)$ to minimize the following infinite-horizon cost function:
	\begin{equation} \label{eq:lqt_cost}
	\begin{aligned}
		&J_\pi = \\
		&\limsup_{T\rightarrow \infty} \frac{1}{T}\bE \left[ \sum_{t=0}^{T-1} \left( (x_t - z_t)^\top Q (x_t - z_t) + u_t^\top R u_t \right) \right],
	\end{aligned}
	\end{equation}
	where $Q \succ 0$ and $R \succ 0$, and $\{z_t \in \bR^n\}$ is the reference trajectory to be tracked.

	The model-based case of LQT assumes that the system dynamics matrices $A$ and $B$ are known, and has been well studied in the literature\cite{lewis2012optimal}. Suppose that the noise process $\{w_t\}$ is white Gaussian with zero mean and bounded covariance, independent of the initial state. The optimal control policy is given by the following linear state-feedback law:
	\begin{equation} \label{eq:policy_structure}
		u_t^* = K^* x_t + K_v^* v_t,
	\end{equation}
	where $K^*$ and $K_v^*$ are gain matrices computed by solving the associated Riccati equations as
	\begin{equation} \label{eq:riccati_equations}
		\begin{aligned}
		P &= Q + A^\top P A + A^\top P B (R + B^\top P B)^{-1} B^\top P A, \\
		K^* &= -(R + B^\top P B)^{-1} B^\top P A, \\
		K_v^* &= (R + B^\top P B)^{-1} B^\top,
		\end{aligned}
	\end{equation}
	and $v_t$ indicates the influence of the future reference trajectory $z_{t}, z_{t+1}, \ldots$ on the current control input $u_t^*$, which can be computed by solving the following backward dynamics:
	\begin{equation} \label{eq:vt_dynamics}
		v_{t} = (A + B K^*)^\top v_{t+1} + Q z_{t}.
	\end{equation}

	\subsection{Indirect certainty-equivalence LQT via model identification}
	In the data-driven setting, the system dynamics matrices $A$ and $B$ are unknown, and need to be estimated from data. The certainty-equivalence approach is to first estimate the system matrices via least squares, and then design the optimal controller based on the estimated model. 

	Consider the states, inputs, noises and the successor states of length $t$:
	\begin{equation} \label{eq:data_collection}
		\begin{aligned}
		X_{0,t} &:= [x_0, x_1, \ldots, x_{t-1}] \in \bR^{n \times t}, \\
		U_{0,t} &:= [u_0, u_1, \ldots, u_{t-1}] \in \bR^{m \times t}, \\
		W_{0,t} &:= [w_0, w_1, \ldots, w_{t-1}]  \in \bR^{n \times t}, \\
		X_{1,t} &:= [x_1, x_2, \ldots, x_{t}] \in \bR^{n \times t},
		\end{aligned}
	\end{equation}
	where the matrices satisfy the system dynamics~\eqref{eq:system_dynamics} in the following matrix form:
	\begin{equation} \label{eq:data_equation}
		X_{1,t} = A X_{0,t} + B U_{0,t} + W_{0,t}.
	\end{equation}
	We assume in this paper that the data is persistently exciting (PE)~\cite{willems2005note}, i.e. $D_{0,t} \triangleq {\left[U_{0,t}^\top, X_{0,t}^\top\right]}^\top$ has   full row rank. The PE condition is necessary in designing LQR, as well as LQT controller~\cite{van2020data,kang2023minimum}, and can be ensured by applying sufficiently rich inputs during data collection.

	Following the certainty-equivalence principle~\cite{dorfler2023certainty}, we can  first estimate the system matrices $(A,B)$ via the least squares estimator and the subspace relation~\eqref{eq:data_equation}:
	\begin{equation} \label{eq:ls_estimator}
		[\hA_t, \hB_t] = \arg\min_{A,B} \Vert X_{1,t} - A X_{0,t} - B U_{0,t} \Vert_F^2 = X_{1,t} D_{0,t}^\dagger.
	\end{equation}
	Then, the certainty-equivalence LQT controller is given by solving the following optimization problem:
	\begin{equation} \label{eq:ce_lqt}
		\begin{aligned}
		\min \limits_{K, K_v}~ &J_\pi, \\
		~s.t. ~~ &\text{\eqref{eq:system_dynamics} with the estimate model $[\hA_t, \hB_t]$,} \\
		&u_t = \pi(x_t,z_t) = K x_t + K_v v_t \\
		&v_{t} = (A + B K) v_{t+1} + Q z_{t}.
		\end{aligned}
	\end{equation}

	Such an indirect data-driven method is also known as the certainty-equivalence (CE) approach, and has been widely studied in the literature~\cite{dean2020sample,mania2019certainty}. In the following, we use $\hK$ and $\hK_v$ to denote the solution of~\eqref{eq:ce_lqt}. 

	\subsection{Direct data-driven LQT via data-based parameterization} \label{sec:direct_lqt}
	Instead of identifying the system model, another line of research focuses on directly learning the optimal controller from data, which is known as the direct data-driven method. 

	Due to the PE condition of data, there exist a matrix $G$ and a matrix $G_v$ such that 
	\begin{equation} \label{eq:data_parameterization}
		\begin{bmatrix}
		K \\ I_n
		\end{bmatrix} =
		\begin{bmatrix}
		U_{0,t} \\ X_{0,t}
		\end{bmatrix} G, ~~
		\begin{bmatrix}
		K_v \\ \boldsymbol{0}_n
		\end{bmatrix} = 
		\begin{bmatrix}
		U_{0,t} \\ X_{0,t}
		\end{bmatrix} G_v
	\end{equation}
	holds for any given $K$ and $K_v$. Then, the closed-loop system under the control policy $u_t = K x_t + K_v v_t$ can be represented by the following data-based dynamics:
	\begin{equation} \label{eq:data_based_dynamics}
		\begin{aligned}
		A+B K &= (X_{1,t} - W_{0,t}) G , \\
		B K_v &= (X_{1,t} - W_{0,t}) G_v.
		\end{aligned}
	\end{equation}
	By the certainty-equivalence principle, the term $W_{0,t}$ in~\eqref{eq:data_based_dynamics} is ignored, and we get a totally data-based representation of the closed-loop system. Then, the direct data-driven LQT controller can be designed by solving the following optimization problem:
	\begin{equation} \label{eq:direct_lqt}
		\begin{aligned}
		\min \limits_{G, G_v}~ &J_\pi, \\
		~s.t. ~~ &x_{t+1} = X_{1,t} G x_t + X_{1,t} G_v v_t, \\
		&u_t = \pi(x_t,z_t) = U_{0,t} G x_t + U_{0,t} G_v v_t, \\
		&v_{t} = X_{1,t} G v_{t+1} + Q z_{t}.
		\end{aligned}
	\end{equation}

	The idea of this data-driven parameterization approach is first proposed in~\citet{de2019formulas} for the LQR problem and only has the homogeneous part, i.e. $K$ and $G$ part in \eqref{eq:data_parameterization}. An extension to the affine policy like~\eqref{eq:policy_structure} is first studied in~\citet{breschi2023data} to the best of our knowledge. However, the parameterization \eqref{eq:data_parameterization} becomes larger as time $t$ increases, which makes the optimization problem \eqref{eq:direct_lqt} intractable when running online.

	\subsection{Direct data-driven LQT with online closed-loop data}

	This paper aims to provide an online direct data-driven policy optimization method for LQT, which extends the DeePO framework \cite{zhao2025policy}. In the following sections we first propose a \emph{reference‑decoupled reformulation} of the LQT problem and a \emph{covariance‑based parameterisation} that yields a fixed‑dimension optimization problem. The resulting algorithms, both offline and online, will then be shown to converge linearly (for the offline case) or to a neighbourhood of optimality (for the online case).

	\section{Direct Data-driven LQT with Covariance Parameterization} \label{sec:covariance_lqt}

	In this section, we first decouple the reference from the policy, then introduce a covariance‑based parameterization that yields a constant‑size optimization problem.

	\subsection{Reference-Decoupled Formulation of LQT}
	Observed from \eqref{eq:riccati_equations}, the optimal control gain is independent of the reference trajectory $\{z_t\}$. Inspired by this fact, we first reformulate the LQT problem into a reference-decoupled form. 

	As a special case, consider the set-point tracking problem, i.e. $z_t \equiv \delta$ for all $t$. By \eqref{eq:vt_dynamics}, $v_t$ becomes also time-invariant with $v_\delta = (I_n - A - B K^*)^{-\top} Q \delta.$ Then, the optimal control policy~\eqref{eq:policy_structure} can be rewritten as
	\begin{equation} \label{eq:setpoint_policy}
		u_t^* = K^* x_t + K_v^*(I_n - A - B K^*)^{-\top} Q \delta \triangleq K^* x_t + L^*\delta.
	\end{equation}
	Now, the set-point tracking problem becomes
	\begin{equation} \label{eq:set_point_lqt}
		\begin{aligned}
		\min \limits_{K, L} ~&J_\pi, \\
		~s.t. ~~ &x_{t+1} = A x_t + B u_t + w_t,\\
		& u_t = \pi(x_t, \delta) = K x_t + L \delta.
		\end{aligned}
	\end{equation}
	Since the previous results on LQT still hold true, $K^*$ in \eqref{eq:setpoint_policy} and $L^* = K_v^*(I - A - B K^*)^{-\top} Q$ must be one of the optimal solutions to \eqref{eq:set_point_lqt}. Intuitively, we can solve the set-point tracking problem~\eqref{eq:set_point_lqt} to get the optimal gain $K^*$ and then compute $K_v^*$ from $L^*$. However, the optimal solution to \eqref{eq:set_point_lqt} is not unique, since all $L$ that satisfy $L\delta = L^* \delta$ are the same policy and thus optimal. Clearly, a single $\delta$ determines $L$ only in its own direction. To Uniquely identify the whole matrix $L$, we need $n$ linearly independent reference points. Specifically, we choose $\delta^i = e^i$, where $e^i$ is the $i$-th column of the identity matrix $I_n$ and obtain the following reformulation of model-based LQT by adding the cost funtions under these $n$ reference points:
	\begin{equation} \label{eq:multi_set_point_lqt}
		\begin{aligned}
		\min \limits_{K, L} ~&J_\pi^{\text{sum}} = \sum_{i=1}^n J_\pi^{(i)}, \\
		~s.t. ~~ &x_{t+1}^{(i)} = A x_t^{(i)} + B u_t^{(i)} + w_t^{(i)},~ i=1,\ldots,n,\\
		& u_t^{(i)} = \pi(x_t^{(i)}, \delta^i) = K x_t^{(i)} + L e^i,
		\end{aligned}
	\end{equation}
	where $J_\pi^{(i)}$ is the cost function~\eqref{eq:lqt_cost} under the reference point $e^i$. 

	The following lemma shows that the optimal solution to~\eqref{eq:multi_set_point_lqt} recovers the optimal gains of the original LQT problem. 
	\begin{lemma} \label{lem:equivalence}
		The optimal solution $(K^{sum}, L^{sum})$ to \eqref{eq:multi_set_point_lqt} is unique, and satisfies that $K^{sum}=K^*$ and $L^{sum} = K_v^*(I_n - A - B K^*)^{-\top} Q$ where $K^*$ and $K_v^*$ are given in \eqref{eq:riccati_equations}.
	\end{lemma}
	\begin{pf}
		By \citet{lewis2012optimal}, the optimal gain $K^i$ to the $i$-th set-point tracking problem, i.e. the problem \eqref{eq:set_point_lqt} with $\delta = e^i$, is unique. Since \eqref{eq:set_point_lqt} is a special case of general LQT, then it is clear that $K^i$ equals to $K^*$ for all $i=1,\ldots,n$. Therefore, the optimal gain $K^{sum}$ to \eqref{eq:multi_set_point_lqt} must also equal to $K^*$.

		For the gain $L^{sum}$, it is also known that all the optimal solutions $L^i$ to the $i$-the set-point tracking problem \eqref{eq:set_point_lqt} will satisfy $L^i \delta = K_v^*(I_n - A - B K^*)^{-\top} Q \delta$. Since $\{\delta^i\}_{i=1}^n$ span the whole space $\bR^n$, then $K_v^*(I_n - A - B K^*)^{-\top} Q$ is the only matrix that minimizes \eqref{eq:multi_set_point_lqt}. This completes the proof. 
	\end{pf}

	For a given policy $\theta = [K, L]$, define some policy-dependent matrices as follows:
	\begin{equation} \label{eq:policy_matrices}
		\begin{aligned}
		P_K = ~& Q + K^\top R K + (A + B K)^\top P_K (A + B K), \\
		Y_K = ~& (I_n - A - B K)^{-1}, \\
		G_\theta =~& Y_K^\top (-Q + K^\top R L + (A+BK)^\top P_K B L). 
		\end{aligned}
	\end{equation}

	The next lemma shows that we can drop the states and inputs in \eqref{eq:multi_set_point_lqt}, and get a compact form that only depends on the policy parameters. 
	\begin{lemma} \label{lem:compact_form}
		The reformulated LQT problem~\eqref{eq:multi_set_point_lqt} is equivalent to the following model-based optimization problem:
		\begin{equation} \label{eq:compact_lqt}
			\begin{aligned}
			\min \limits_{\theta} ~& C(\theta) = \mt (Q + L^\top R L + L^\top B^\top P_K B L\\ 
			&~~~~~~~~~ + nP_K+ 2G_\theta^\top B L), \\
			~s.t. ~~& \rho(A+BK) <~ 1,
			\end{aligned}
		\end{equation}
	\end{lemma}

	\begin{remark}
	All the results in this subsection hold true if we replace the model-based matrices $A$ and $B$ by their estimates $\hA$ and $\hB$. In the following, we use the notation with a hat to denote the estimated quantities in this subsection. 
	\end{remark}

	\subsection{Covariance Parameterization of LQT}
	a data-driven version of~\eqref{eq:compact_lqt} is developed using only the collected data. To this end, a data-based parameterization like that introduced in Section~\ref{sec:direct_lqt} is needed to express the cost function $C(\theta)$. To address the increasing dimensionality in the online setting, we extend the covariance parameterization in \cite{zhao2025data} for LQR to the LQT case. 

	First, define the covariance matrices as $\Lambda_t := \frac{1}{t} D_{0,t} D_{0,t}^\top$
	and define the verified data matrices as
	\begin{equation} \label{eq:verified_data_matrices}
		\begin{aligned}
		\oX_{0,t} &:= \frac{1}{t} X_{0,t} D_{0,t}^\top, \quad \oU_{0,t} := \frac{1}{t} U_{0,t} D_{0,t}^\top, \\
		\oX_{1,t} &:= \frac{1}{t} X_{1,t} D_{0,t}^\top, \quad \oW_{0,t} := \frac{1}{t} W_{0,t} D_{0,t}^\top.
		\end{aligned}
	\end{equation}
	Then, the covariance parameterization of LQT is given as $\xi = [V,H]$ with $V,H \in \bR^{(m+n) \times n}$ satisfying
	\begin{equation} \label{eq:covariance_parameterization}
		\begin{aligned}
		\begin{bmatrix} K \\ I_n \end{bmatrix} = \Lambda_t V = \begin{bmatrix} \oU_{0,t} \\ \oX_{0,t} \end{bmatrix} V, 
		\begin{bmatrix} L \\ \boldsymbol{0}_n \end{bmatrix} = \Lambda_t H = \begin{bmatrix} \oU_{0,t} \\ \oX_{0,t} \end{bmatrix} H.
		\end{aligned}
	\end{equation}

	By the covariance parameterization \eqref{eq:covariance_parameterization}, the closed-loop system matrices can be represented directly by the data matrices and policy $\xi$ as
	\begin{equation} \label{eq:closed_loop_matrices}
		\begin{aligned}
		A + B K &= (\oX_{1,t} - \oW_{0,t}) V, \\
		B L &= (\oX_{1,t} - \oW_{0,t}) H.
		\end{aligned}
	\end{equation}
	Then, by substituting \eqref{eq:closed_loop_matrices} into \eqref{eq:policy_matrices} and \eqref{eq:compact_lqt}, we can obtain a data-based representation of the LQT problem:
	\begin{equation} \label{eq:data_based_lqt} 
		\begin{aligned}
		\min \limits_{\xi} ~& C(\xi) = \mt (Q + H^\top \oU_{0,t}^\top R \oU_{0,t}H + H^\top \oX_{1,t}^\top P_V \oX_{1,t} H \\ 
		&~~~~~~~~~ + nP_V+ 2G_\xi^\top \oX_{1,t} H), \\
		~s.t. ~~& \rho(\oX_{1,t}V) <~ 1, \oX_{0,t} V = \boldsymbol{I}_n, \text{ and }\oX_{0,t} H = \boldsymbol{0}_n,
		\end{aligned}
	\end{equation}
	where $P_V$ and $G_\xi$ are defined as:
	\begin{equation}\label{eq:data_based_policy_matrices}
		\begin{aligned}
		P_V = ~& Q + V^\top \oU_{0,t}^\top R \oU_{0,t}  V  + V^\top\oX_{1,t}^\top P_V \oX_{1,t}V, \\
		Y_V = ~& (I_n - \oX_{1,t} V)^{-1}, \\
		G_\xi =~& Y_V^\top (-Q + V^\top \oU_{0,t}^\top R \oU_{0,t}H  + V^\top \oX_{1,t}^\top P_V \oX_{1,t} H). 
		\end{aligned}
	\end{equation}
	Notice that the known noises $\oW_{0,t}$ are omited by the certainty-equivalence principle~\cite{dorfler2023certainty}.

	Besides the fixed size of the parameters, the advantage to use the covariance parameterization can also be shown by the equivalence between~\eqref{eq:data_based_lqt} and~\eqref{eq:ce_lqt}. Let $\xi^* = [V^*~H^*]$ denote the optimal solution to \eqref{eq:data_based_lqt}, then we have the following theorem.
	\begin{theo}\label{thm:equivalence}
		The optimal solution to~\eqref{eq:data_based_lqt} and~\eqref{eq:ce_lqt} coincide under the change of variables \eqref{eq:covariance_parameterization} and $\hL = \hK_v(I_n - \hA - \hB \hK)^{-\top} Q$.
	\end{theo}

	Theorem \ref{thm:equivalence} indicates that the covariance parameterization based data-driven LQT~\eqref{eq:data_based_lqt} is self-regularized~\cite{zhao2025regularization}, in the sence that it does not need an extra regularization term to guarantee equivalence like the parameterization in Section~\ref{sec:direct_lqt} \cite{dorfler2022role}. Note that the optimal costs of \eqref{eq:ce_lqt} and \eqref{eq:data_based_lqt} do not coincide in general like in \cite{zhao2025data}, since the reference-decoupled LQT formulation \eqref{eq:multi_set_point_lqt} omits the reference signal. However, by Lemma~\ref{lem:equivalence}, the optimal solution to \eqref{eq:data_based_lqt} can still recover the optimal gains of the original LQT problem.

	\section{Offline DeePO for LQT}\label{sec:offline_deepo}
	In this section, we present an offline DeePO algorithm to solve the data-driven LQT problem~\eqref{eq:data_based_lqt}. We omit the time index $t$ in the rest of this section for simplicity.

	\subsection{The offline DeePO Algorithm for solving~\eqref{eq:data_based_lqt}}
	We adopt the policy gradient method to solve the data-driven LQT problem~\eqref{eq:data_based_lqt}. Define the feasible policy set by
	\begin{equation*}
		\cS = \{\xi = [V,H] ~|~\oX_0 V = \boldsymbol{I}_n, \oX_0 H = \boldsymbol{0}_n, \rho(\oX_{1} V) < 1 \}.
	\end{equation*}
	Besides, for any policy $\xi \in \cS$, define the following policy-dependent matrices:
	\begin{equation} \label{eq:policy_dependent_matrices}
		\begin{aligned}
		E_V = ~& (\oU_0^\top R \oU_0 + \oX_1^\top P_V \oX_1) V, \\
		F_\xi = ~& \oX_1^\top G_\xi + \oU_0^\top R \oU_0 H + \oX_1^\top P_V \oX_1 H.\\
		\Sigma_V = ~& I_n + \oX_1 V \Sigma_V V^\top \oX_1^\top,\\
		Z_\xi = ~& Y_V \oX_1 H, \\
		\Phi_\xi = ~& \begin{bmatrix}n\Sigma_V + Z_\xi Z_\xi^\top & Z_\xi \\ Z_\xi^\top & I_n \end{bmatrix}
		\end{aligned}
	\end{equation}

	Then, the data-driven policy gradient $\nabla_\xi C(\xi)$ can be computed as follows.
	\begin{lemma} \label{lem:policy_gradient}
		For any policy $\xi \in \cS$, the policy gradient of the cost function $C(\xi)$ in \eqref{eq:data_based_lqt} is given by
		\begin{equation*}
			\nabla_\xi C(\xi) = 2 \begin{bmatrix} E_V & F_\xi \end{bmatrix} \Phi_\xi.
		\end{equation*}
	\end{lemma}

	Using Lemma~\ref{lem:policy_gradient}, we can apply the gradient descent method to solve \eqref{eq:data_based_lqt}. Starting from an initial policy $\xi^0 \in \cS$, the policy is updated iteratively as
	\begin{equation} \label{eq:policy_update}
		\xi^+ = \xi - \eta \Pi_{\oX_0} \nabla_\xi C(\xi),
	\end{equation}
	where $\eta > 0$ is the step size, $\Pi_{\oX_0} = I_n - \oX_0^\dagger \oX_0$ is the projection matrix to ensure that the updated policy $\xi^{k+1}$ satisfies the affine constraints in $\cS$, with $\oX_0^\dagger$ being the Moore-Penrose pseudoinverse of $\oX_0$. The equation \eqref{eq:policy_update} forms our offline DeePO algorithm for solving the data-driven LQT problem~\eqref{eq:data_based_lqt}. 

	\subsection{Global Linear Convergence of the Offline DeePO}
	Since the cost function $C(\xi)$  in~\eqref{eq:data_based_lqt} is non-convex, it is not straightforward to analyze the convergence of the offline DeePO algorithm in~\eqref{eq:policy_update}. However, by exploiting the gradient domination and local smoothness property of the model-based LQT problem~\eqref{eq:compact_lqt} with estimate model $(\hA, \hB)$, it is sufficient to establish the global linear convergence of a model-based policy gradient algorithm for solving~\eqref{eq:compact_lqt}. Then, using the equivalence between the above model-based algorithm with~\eqref{eq:policy_update}, we can establish the global linear convergence of the offline DeePO algorithm. 

	First, we show that the DeePO iteration \eqref{eq:policy_update} is equivalent to a model-based policy one in the following lemma.
	\begin{lemma} \label{lem:equivalence_deepo_model_based}
		The policy update in \eqref{eq:policy_update} is equivalent to the following model-based policy gradient update:
		\begin{equation} \label{eq:model_based_policy_update}
			\theta^{k+1} = \theta^k - \eta M \nabla_\theta \hC(\theta^k),
		\end{equation}
		where $M = \oU_0 (I_n-\oX_0^\dagger \oX_0) \oU_0^\top$ is a data-dependent scaling matrix. Furthermore, $M$ is positive definite. 
	\end{lemma}

	Using Lemma~\ref{lem:equivalence_deepo_model_based}, the convergence analysis of the offline DeePO algorithm can be transformed to that of the model-based policy gradient method for solving \eqref{eq:compact_lqt} with $(\hA, \hB)$. Although the problem \eqref{eq:compact_lqt} is still non-convex, we can show that it satisfies the gradient domination and local smoothness properties, which are standard to establish the global linear convergence of the policy gradient method~\cite{fazel2018global}.

	Let $\hC^*$ be the optimal cost of \eqref{eq:data_based_lqt} and \eqref{eq:compact_lqt} with $(\hA, \hB)$, and define 
	\begin{equation*}
		\begin{aligned}
		\widehat{\cS}_{\theta} &= \{\theta = [K,L] ~|~ \rho(\hA + \hB K) < 1 \}, \\
		\widehat{\cS}_{\theta} (a) &= \{\theta = [K,L] ~|~ \rho(\hA + \hB K) < 1, \hC(\theta) \leq a \}. \\
		\end{aligned}
	\end{equation*}
	We give the gradient domination and local smoothness properties of the model-based LQT problem~\eqref{eq:compact_lqt}.
	\begin{lemma}\label{lem:gradient_domination_smoothness}
		For any stabilizing policy $\theta \in \widehat{\cS}_{\theta} (a)$ and any $a > \hC^*$, the following properties hold:
		\begin{itemize}
			\item (Coercivity) The cost function $\hC(\theta)$ is coercive in $\theta$, i.e., $\hC(\theta) \to \infty$ as $\theta \to \partial \widehat{\cS}_{\theta}$ where $\partial$ denotes the boundary of the set. Moreover, $\widehat{\cS}_{\theta} (a)$ is a compact set.
			\item (Gradient Domination) There exists a constant $\widehat{\mu}(a) > 0$ such that
			\begin{equation*}
				\hC(\theta) - \hC^* \leq \widehat{\mu}(a) \Vert \nabla_\theta \hC(\theta) \Vert_F^2,
			\end{equation*}
			where $\widehat{\theta}^*$ is the optimal policy to \eqref{eq:compact_lqt}.
			\item (Local Smoothness) For any $\theta', \theta \in \widehat{\cS}_{\theta} (a)$ satisfying $\theta + b(\theta' - \theta) \in \widehat{\cS}_{\theta} (a), \forall b \in [0,1]$, there exists a constant $\widehat{L}(a) > 0$ which is polynomial in $a$ such that
			\begin{equation*}
				\hC(\theta') \leq \hC(\theta) + \langle \nabla_\theta \hC(\theta), \theta' - \theta \rangle + \frac{\widehat{L}(a)}{2} \Vert \theta' - \theta \Vert_F^2.
			\end{equation*}
		\end{itemize}
	\end{lemma}

	Based on Lemma~\ref{lem:gradient_domination_smoothness}, we can establish the global linear convergence of the model-based PO \eqref{eq:model_based_policy_update}. Different from the case in \cite{fazel2018global} for the LQR problem, the smoothness and gradient domination properties for LQT only holds locally. Therefore, we need to choose a small enough step size to ensure that the policy sequence generated by \eqref{eq:model_based_policy_update} always stays in a compact set. For some initial policy $\theta^0 \in \widehat{\cS}_{\theta}$, define $\hL^0 = \hL(\hC(\theta^0))$ and $\widehat{\mu}^0 = \widehat{\mu}(\hC(\theta^0))$. Then we have the following convergence result.
	\begin{lemma} \label{lem:convergence_model_based_po}
		For any initial policy $\theta^0 \in \widehat{\cS}_{\theta}$ and stepsize $\eta \in \mathopen( 0, 1/(\hL^0 \Vert M \Vert) \mathclose]$, the model-based PO algorithm in \eqref{eq:model_based_policy_update} leads to $\theta^k \in \widehat{\cS}_{\theta}(\hC(\theta^0))$ for all $k$. Moreover, $\theta^k$ converges linearly as:
		\begin{equation*}
			\hC(\theta^{k+1}) - \hC^* \leq (1 - \nu) (\hC(\theta^k) - \hC^*),
		\end{equation*}
		where $\nu = \eta(1-\hL^0 \eta \Vert M \Vert/2)\us(M)/\widehat{\mu}^0$.
	\end{lemma}

	As a direct consequence of Lemmas~\ref{lem:equivalence_deepo_model_based} and~\ref{lem:convergence_model_based_po}, we can establish the global linear convergence of the offline DeePO algorithm in \eqref{eq:policy_update}.
	\begin{theo} \label{thm:convergence_deepo}
		For any initial policy $\xi^0 \in \cS$ and stepsize $\eta \in \mathopen( 0, 1/(\hL^0 \Vert M \Vert) \mathclose]$, the offline DeePO algorithm in \eqref{eq:policy_update}  converges linearly as:
		\begin{equation*}
			C(\xi^{k+1}) - \hC^* \leq (1 - \nu) (C(\xi^k) - \hC^*).
		\end{equation*}
	\end{theo}

	\section{Online DeePO for LQT}\label{sec:online_deepo}
	In this section, we extend the offline DeePO algorithm in Section \ref{sec:offline_deepo} to an online version by using the data collected from the closed-loop system under the current policy. We first present the algorithm and then prove that, with appropriate conditions, the online DeePO algorithm converges linearly to a neighborhood of the optimal policy.

	\subsection{Algorithm Design} \label{sec:online_algorithm}

	The objective of the online adaptive algorithm is to update the control policy concurrently with system operation, ensuring closed-loop stability while steering the policy towards the global optimum. Unlike the offline scenario, online adaptation inherently requires a delicate balance between persistent exploration and optimal exploitation. The complete procedure is summarized in Algorithm \ref{alg:adaptive_deepo_dynamic}.

	\begin{algorithm}[ht]
		\caption{Online DeePO for dynamic reference tracking}
		\label{alg:adaptive_deepo_dynamic}
		\begin{algorithmic}[1]
			\STATE \textbf{Input:} Initial policy parameters $[K_{t_0}, K_{v,t_0}]$, step size $\eta$, initiate data $X_{0,t_0}, U_{0,t_0}, X_{1,t_0}$, reference trajectory $z_t$ for $t=t_0,t_0+1,t_0+2,\cdots$;
			\STATE Calculate the data-based parameterization $\xi_{t_0}$ via
			\[V_{t_0}' = \Lambda^{-1}_{t_0} \begin{bmatrix} K_{t_0} \\ I_n \end{bmatrix},\]
			\FOR{$t=t_0,t_0+1,t_0+2,\cdots$}
			\STATE Calculate the reference-related state via
			\[v_t = \oX_{1,t} V_t' v_{t+1} + Q z_t;\]
			\STATE Apply the control input $u_t = K_t x_t + K_{v,t} v_t + e_t$ and obtain the next state $x_{t+1}$;
			\STATE Calculate the data-based parameterization $\xi_t$ via
			\begin{align*}
			V_{t+1} &= \Lambda^{-1}_{t+1} \begin{bmatrix} K_{t} \\ I_n \end{bmatrix},\\
			H_{t+1} &= \Lambda^{-1}_{t+1} \begin{bmatrix} K_{v,t}(I_n - \oX_{1,t} V_{t}')^{-\top}Q \\ \boldsymbol{0}_n \end{bmatrix};
			\end{align*}
			\STATE Update the policy parameters via
			\[ \xi_{t+1}' = \xi_{t+1} - \eta \Pi_{\oX_{0,t+1}} \nabla_\xi C_{t+1}(\xi_{t+1})\]
			where $\nabla_\xi C_{t+1}(\xi_{t+1})$ is calculated directly from $(X_{0,t+1}, U_{0,t+1}, X_{1,t+1})$ via Lemma~\ref{lem:policy_gradient};
			\STATE Update the policy gains via
			\begin{align*}
				K_{t+1} &= \oU_{0,t+1} V_{t+1}', \\
				K_{v,t+1} &= \oU_{0,t+1} H_{t+1}' (I_n - \oX_{1,t+1} V_{t+1}')^\top Q^{-1};
			\end{align*}
			\ENDFOR
		\end{algorithmic}
	\end{algorithm}

	The core challenge in transitioning to an online framework lies in the time-varying nature of the data matrices $D_{0,t}$. Consequently, the covariance parameterization $\xi_t$ in \eqref{eq:covariance_parameterization} shifts with time, even if the underlying model-based policy $\theta_t = [K_t, K_{v,t}]$ remains fixed. Optimizing $\xi_t$ directly across different time steps is therefore ill-posed. To circumvent this, Algorithm \ref{alg:adaptive_deepo_dynamic} employs a dynamic re-parameterization strategy. Between consecutive steps $t$ and $t+1$, the current policy is mapped back to the data-independent space $\theta_t$. Once new data $(x_{t+1}, u_t)$ is incorporated, the policy is re-parameterized into $\xi_t$ (Line 6) using the updated data matrices, enabling a mathematically consistent projected gradient update (Line 7).

	At each time step $t$, the control input (Line 5) is synthesized using two components: the exploitation term $K_t x_t + K_{v,t} v_t$, which drives the tracking performance, and the exploration noise $e_t$. The injection of $e_t$ is strictly necessary to maintain the persistent excitation of the online data, a prerequisite for algorithmic convergence that will be rigorously analyzed in Section \ref{sec:convergence_analysis}.

	Implementing the exploitation term requires computing the non-homogeneous tracking state $v_t$. Since the true closed-loop dynamics $A+BK_t$ are inaccessible, $v_t$ is approximated using the data-based equivalent $\oX_{1,t} V_t'$ (Line 4). Theoretically, the exact evaluation of \eqref{eq:vt_dynamics} demands future reference knowledge $\{z_k\}_{k=t}^\infty$. The convergence of $v_t$ is guaranteed by the stability of the data-based dynamics $\oX_{1,t} V_t'$, provided that $z_t$ is bounded. In practical implementations where an infinite-horizon preview is unfeasible, we adopt a finite-horizon truncation. By acquiring a preview of $N+1$ steps, we initialize $v_{t+N} = (I - \oX_{1,t} V_t')^{-1} z_{t+N}$ and iterate backward to obtain $v_t$. Due to the closed-loop stability, the truncation error decays exponentially with respect to $N$. In Section \ref{sec:experiment}, a modest preview horizon of $N=10$ is chosen for online control.

	\subsection{Convergence Analysis} \label{sec:convergence_analysis}

	To establish the convergence of Algorithm \ref{alg:adaptive_deepo_dynamic}, we introduce necessary assumptions regarding the system noise and the persistent excitation of the input data.

	\begin{assumption} \label{asm:noise}
		The process noise $w_t$ has a bounded covariance, i.e., there exists $\delta_t \geq 0$ such that $$\frac{1}{t}\Vert W_{0,t} W_{0,t}^\top \Vert \leq \delta_t^2.$$
	\end{assumption}

	\begin{assumption} \label{asm:input}
		The control input is persistently exciting, i.e., there exists $\gamma_t > 0$ such that $\us(\Lambda_t) \geq \gamma_t^2$.
	\end{assumption}

	Assumption \ref{asm:noise} bounds the noise variance rather than specifying its distribution. This assumption is standard in data-driven control to ensure the empirical data encapsulates reliable, deterministic system information. Specifically, strictly bounded noise yields $\delta_t \sim \cO(1)$, whereas Gaussian noise typically satisfies $\delta_t \sim \cO(1/\sqrt{t})$ with high probability.

	Assumption \ref{asm:input} formally states the PE condition, which guarantees the data is sufficiently informative. While this condition mathematically involves both the input $u_t$ and the state $x_t$, it can be inherently satisfied by enforcing PE strictly on the input \cite{coulson2023quantitative}. Practically, this is achieved by injecting the exploration noise $e_t$. Although rigorously designing deterministic, bounded PE sequences remains an open challenge, a widely adopted empirical strategy is to sample $e_t$ from a Gaussian distribution. Our experiments, c.f. Section \ref{sec:experiment}, demonstrate that under Gaussian excitation, the minimum singular value $\us(\oU_{0,t})$ grows at a rate of $\cO(\sqrt{t})$ with high probability, corresponding to a unit excitation of $\gamma_t \sim \cO(1)$. This aligns with established results in random matrix theory for matrices with independent entries \cite{tao2023topics}. However, extending these exact theoretical bounds to the correlated closed-loop sequences inherent in online control remains an active area of research.

	Together, these assumptions quantify the data quality, prompting the definition of the data ``Signal-to-Noise Ratio" (SNR):
	$$\text{SNR}_t := \frac{\gamma_t}{\delta_t}.$$
	Consequently, if the excitation $e_t$ is a unit excitation and the process noise $w_t$ is Gaussian, the SNR grows at $\text{SNR}_t \sim \cO(\sqrt{t})$ with high probability.

	We define the feasible policy sets for the true system $(A,B)$ as:
	\begin{equation*}
		\begin{aligned}
			\cS_{\theta} &= \{\theta = [K,L] ~|~ \rho(A + B K) < 1 \}, \\
			\cS_{\theta} (a) &= \{\theta = [K,L] ~|~ \rho(A + B K) < 1, C(\theta) \leq a \}.
		\end{aligned}
	\end{equation*}

	Based on these feasible sets, we establish the initial assumption on the policy, followed by a boundedness assumption on the tracking reference.
	\begin{assumption}
		The initial policy $\theta_{t_0}$ is internally stabilizing, i.e., $\theta_{t_0}\in \cS_{\theta}$.
	\end{assumption}
	\begin{assumption}
		For all $t\geq t_0$, the tracking reference satisfies $\Vert z_t \Vert \leq \overline{z}$.
	\end{assumption}

	Under the aforementioned assumptions, the following theorem establishes the convergence properties of the online DeePO algorithm for time-varying signal tracking.
	\begin{theo} \label{thm:convergence_online_deepo}
		There exist constants $\nu_i>0, i\in\{1,2,3,4\}$ and $\overline{\mu}$ depending on $A,B,Q,R,\theta_{t_0}$ such that: if for all $t$, $\text{SNR}_t \geq \text{max}\{\nu_1,\nu_2 \Vert M_t \Vert/\us(M_t)\}$ and $\eta_t \leq \Vert M_t\Vert^{-1}\nu_3$, then the policy obtained by Algorithm \ref{alg:adaptive_deepo_dynamic} satisfies:
		\begin{equation*}
				\begin{aligned}
				C(\theta_t) - C^* &\le  \prod_{i=t_0+1}^{t} \left( 1 - \frac{\eta_i \underline{\sigma}(M_i)}{2\overline{\mu}} \right) (C(\theta_{t_0}) - C^*) \\
				&+ \nu_4 \sum_{i=t_0+1}^{t} \prod_{s=i+1}^{t} \left( 1 - \frac{\eta_s \underline{\sigma}(M_s)}{2\overline{\mu}} \right) \frac{1}{\text{SNR}_i},
				\end{aligned}
		\end{equation*}
		where $M_t = \oU_{0,t} (I-\oX_{0,t}^\dagger \oX_{0,t}) \oU_{0,t}^\top$.
	\end{theo}

	Theorem \ref{thm:convergence_online_deepo} reveals that the tracking error decays through two competing mechanisms. The first term dictates a linear convergence trend driven by the policy gradient, provided Gaussian excitation and noise yield $\us(M_t) \sim \cO(1)$ with high probability. The second term, bottlenecked by the SNR, accounts for the accumulation of estimation errors. Since Gaussian assumptions yield an SNR growing at $\text{SNR}_t\sim \cO(\sqrt{t})$, this second term decays sublinearly. Consequently, the overall algorithm converges linearly to a residual neighborhood of the optimal policy, the size of which is fundamentally constrained by the SNR.

	A critical prerequisite for Theorem \ref{thm:convergence_online_deepo} is that $\text{SNR}_t$ bounds the condition number of $M_t$, and the step size $\eta_t$ scales inversely with $\Vert M_t \Vert$. Since Lemma \ref{lem:equivalence_deepo_model_based_2} establishes $\us(M_t) \geq \gamma_t^4$, adopting a fixed step size under unit excitation $\gamma_t \sim \cO(1)$ rigorously necessitates a uniform upper bound on $\Vert M_t \Vert$. We establish this by leveraging sequential stability analysis to bound the state $x_t$.

	\begin{theo} \label{thm:sequential_stability}
		There exists a constant $\overline{p}_7$ dependent on $\overline{C}$ such that: if
		\[ 
		\begin{aligned}
		&\frac{\delta_t}{\gamma_t} \leq \text{min}\left\{\underline{p}_1, \frac{\us(M_t)}{2\overline{p}_2 \overline{p}_5 \overline{\mu} \Vert M_t \Vert}\right\} \\
		&\eta_t \leq \frac{1}{\Vert M_t \Vert}\cdot \text{min}\left\{\frac{\underline{p}_3 \gamma_t}{\overline{p}_2 \delta_t}, \underline{p}_6, \frac{\underline{p}_3}{\overline{p}_7}, \frac{\underline{a}}{4\overline{\kappa}^2\overline{p}_4\overline{p}_7}, \frac{1}{2\overline{p}_4 \overline{p}_7} \right\},
		\end{aligned}
		\]
		then it holds that
		\[
		\Vert x_t \Vert \leq \overline{\kappa}(1-\frac{\underline{a}}{2})^{t-t_0}\Vert x_{t_0} \Vert + \frac{2\overline{\kappa}}{\underline{a}}\max_{t_0 \leq i <t} \Vert Be_i + B L v_i + w_i \Vert,\]
		where
		\[ \overline{\kappa} = \sqrt{\frac{\overline{C}}{\text{min}\{\us(R),\us(Q)\}}}, \underline{a} = 1-\sqrt{1-\frac{1}{\overline{\kappa}^2}}.
		\]
	\end{theo}

	The proof of Theorem \ref{thm:sequential_stability} follows the same vein as the proof in \cite{zhao2025policy}. Since the reference $z_t$ (and consequently $v_t$) is bounded, Theorem \ref{thm:sequential_stability} guarantees that $\Vert x_t \Vert$ remains bounded over time. If $e_t$ and $w_t$ are strictly bounded, a deterministic upper bound exists; if they are Gaussian, a high-probability upper bound can be established. Because the control input is $u_t = K x_t + L v_t + e_t$, and $M_t$ is constructed from these historical states and inputs, the uniform boundedness of $x_t$ directly implies that $\Vert M_t \Vert$ is bounded. In Lemma \ref{lem:equivalence_deepo_model_based_2} in the Appendix, we also establish a lower bound on $\us(M_t)$ in terms of the excitation $\gamma_t$. By substituting this supremum for $\Vert M_t \Vert$ and replacing $\us(M_t)$ with $\gamma_t^4$ in the conditions of Theorem \ref{thm:convergence_online_deepo}, we obtain a conservative, matrix-independent criterion that theoretically justifies the selection of a valid, fixed step size throughout the online adaptation.

	\begin{remark} \label{rmk:different_eta}
		It follows from the proof of Theorem \ref{thm:convergence_online_deepo} that the convergence of Algorithm \ref{alg:adaptive_deepo_dynamic} is preserved if the update direction is preconditioned by an arbitrary positive definite matrix $D$, provided the step size is appropriately adjusted. As a special case, we can choose different step sizes for the $V$- and $H$-components of the gradient to maximize the optimization efficiency. This is validated in Section \ref{sec:experiment}.
	\end{remark}

	\section{Numerical Experiments} \label{sec:experiment}

	In this section, we validate the theoretical properties of the proposed offline and online DeePO algorithms. We consider a randomly generated, open-loop stable discrete-time linear system. The open-loop stability allows us to initialize the algorithms with a zero policy ($\theta_0 = 0$) without loss of generality. The system matrix $A$ is given by:
	$$A = \begin{bmatrix}
		-0.229& 0.247& -0.511& 0.493\\
		0.846& 0.159& 0.722& 0.529\\
		-0.018& 0.07& 0.3& 0.758\\
		0.247& 0.546& -0.511& -0.176
	\end{bmatrix}.
	$$
	To evaluate the algorithms under different controllability conditions, we define a fully actuated input matrix $B_{\text{full}}$:
	$$B_{\text{full}} = \begin{bmatrix}
		-0.633 & 0.938 & 0.132 & -0.527\\
		0.262& -0.796& 0.264& -0.350\\
		0.461& -0.180& -0.428& 0.457\\
		0.774& 0.112& -0.285& -0.168
	\end{bmatrix},
	$$
	and an underactuated matrix $B_{\text{under}}$ formed by extracting the first two columns of $B_{\text{full}}$.

	For both offline and online scenario, we pre-collect $T=9$ samples generated with Gaussian exploration noise $\mathcal{N}(0, I_4)$ and process noise $w_t \sim \mathcal{N}(0, 0.1^2 I_4)$. We set the weight matrices to $Q = I_4$ and $R = 0.01 I_2$, and the base learning rate to $\eta = 0.01$. 

	Figure \ref{fig:oc} illustrates the convergence trajectories of the offline DeePO algorithm for both fully actuated and underactuated configurations. The results strictly align with Theorem \ref{thm:convergence_deepo}, demonstrating a global linear convergence rate for both the cost function and the policy parameters. 

	\begin{figure}[t!]
		\centering
		\includegraphics[width = 0.47\textwidth]{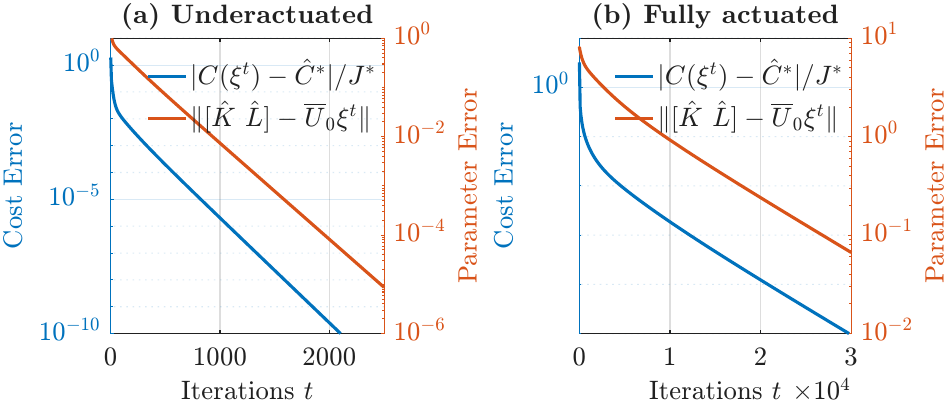}
		\caption{Convergence of the Offline DeePO algorithm for LQT.}
		\label{fig:oc}
	\end{figure}

	\begin{figure}[t!]
		\centering
		\includegraphics[width = 0.47\textwidth]{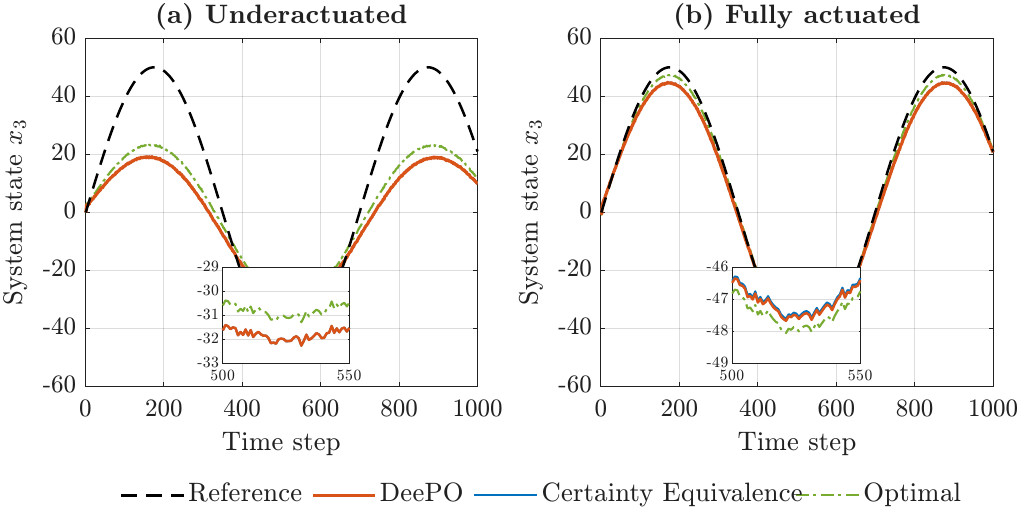}
		\caption{Tracking performance of the Offline DeePO algorithm for LQT.}
		\label{fig:ofp}
	\end{figure}

	\begin{figure}[t!]
		\centering
		\includegraphics[width = 0.47\textwidth]{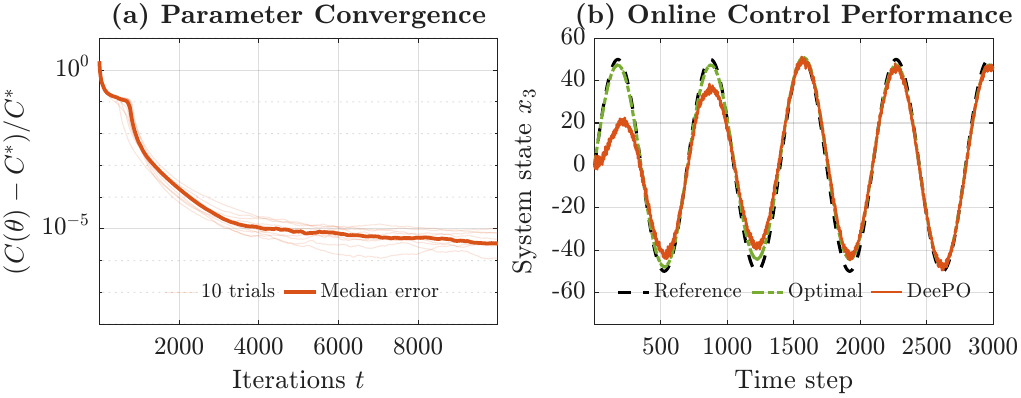}
		\caption{Convergence andTracking performance of the Online DeePO algorithm for LQT.}
		\label{fig:onp}
	\end{figure}

	\begin{figure}[t!]
		\centering
		\includegraphics[width = 0.47\textwidth]{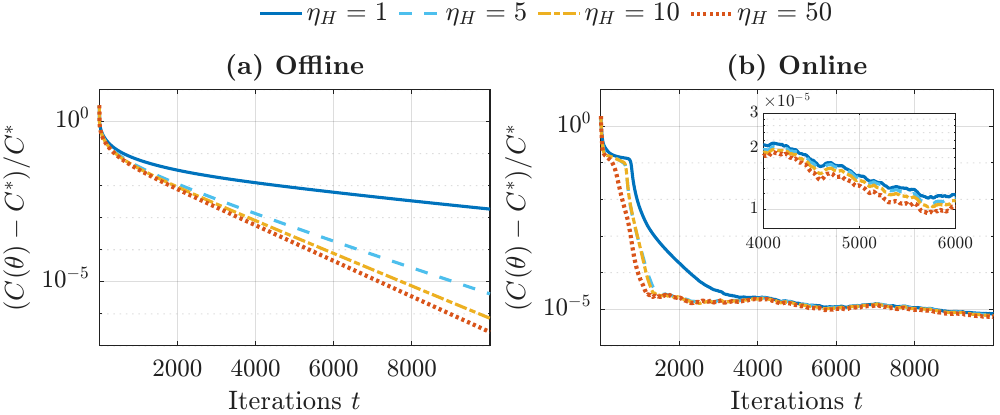}
		\caption{Convergence with different step sizes for the $H$-component of the gradient.}
		\label{fig:ch}
	\end{figure}

	\begin{figure}[t!]
		\centering
		\includegraphics[width = 0.47\textwidth]{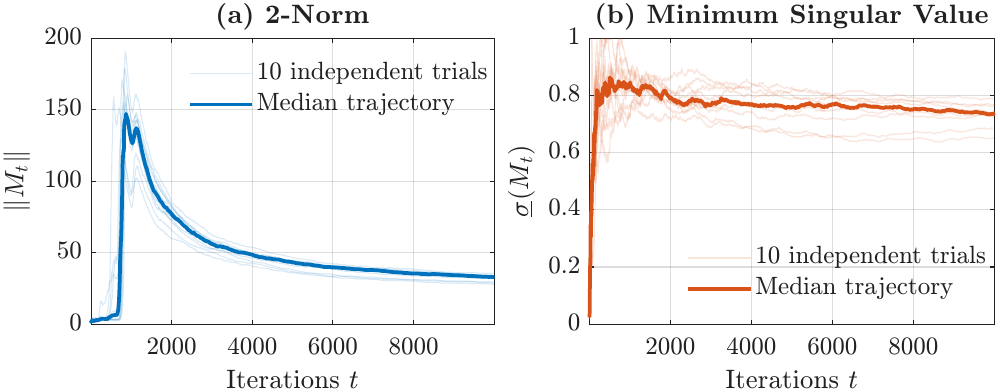}
		\caption{Evolution of $\underline{\sigma}(\oU_{0,t})$ and $\underline{\sigma}(M_t)$ during the online adaptation.}
		\label{fig:sc}
	\end{figure}

	To evaluate tracking performance, the converged policy is deployed for 3000 steps to track a time-varying reference: $z_t = [50 \sin(0.003t), 0.003t, 50 \sin(0.009t), 10]^\top$. As shown in Fig. \ref{fig:ofp}, the offline DeePO policy achieves a tracking cost almost identical to the certainty-equivalence (CE) baseline and closely approximates the global optimal performance. Predictably, perfect tracking is unattainable for the underactuated system regardless of the policy, manifesting as inherent residual errors.

	Since in the offline mode, the differences between underactuated and fully actuated have been fully demonstrated, in the online mode we only present the results of fully actuated. The results, averaged over 10 independent Monte Carlo runs, are presented in Fig. \ref{fig:onp}.
	The perfectly captures the two-phase convergence predicted by Theorem \ref{thm:convergence_online_deepo}. The policy initially converges at a linear rate before transitioning to a sublinear rate, ultimately bounded by the accumulating SNR bottleneck. Consequently, the real-time tracking performance dynamically improves as the policy adapts, eventually aligning with the optimal tracking trajectory.

	In Remark \ref{rmk:different_eta}, we discussed the possibility of using different step sizes for the $V$- and $H$-components of the gradient. We empirically validate this by scaling the base learning rate $\eta$ by factors of $\eta_H = 1, 5, 10 \text{ and }50$ for the $H$-component while keeping it fixed for the $V$-component. As shown in Fig. \ref{fig:ch}, a larger step size for the $H$-component leads to faster convergence in the early stages. However, for the online mode, the final convergence rate is bottlenecked by the SNR, and a larger step size does not lead to better final performance.

	A critical pillar of our convergence is the boundedness of $\Vert M_t \Vert$ and $\us(M_t)$ under Gaussian excitation. To empirically validate these assumptions, we track the evolution of $\underline{\sigma}(\oU_{0,t})$ and $\underline{\sigma}(M_t)$ throughout the online adaptation process, shown in Fig. \ref{fig:sc}. The results confirm that both $\Vert M_t \Vert$ and $\underline{\sigma}(M_t)$ exhibit a trend consistent with $\cO(1)$, thereby validating the theoretical conditions for convergence.

	\section{Conclusion} \label{sec:conclusion}

	This paper presents a direct data-driven policy optimization framework for the LQT problem. By introducing a reference-decoupled reformulation and extending covariance parameterization, we ensure a constant policy parameter dimension independent of data length. We develope offline and online DeePO algorithms, establishing global linear convergence for the offline case and proving that the online adaptive policy converges to an optimality neighborhood bounded by the inverse SNR. Numerical results validate the effectiveness of our proposed methods in dynamic tracking tasks. 
	
	Future research will focus on establishing theoretical guarantees for persistent excitation (PE) under Gaussian noise and exploring more efficient exploration mechanisms. Additionally, investigating the impact of input design on the properties of the scaling matrix $M$, such as its condition number, remains a promising direction to further bridge the gap between data-driven and model-based methodologies.

	\appendix

	\section{Proof of Lemma~\ref{lem:compact_form}} \label{app:proof_lem_compact}
		For the $i$-th set-point tracking problem in \eqref{eq:multi_set_point_lqt}, define an advantage function as 
		\begin{equation*}
			\begin{aligned}
				f^i(x;\theta) = &\bE \bigg \{ \sum_{t=0}^\infty (x_t - e^i)^\top Q (x_t - e^i) + u_t^\top R u_t \\
				& - C^i(\theta) | x_0 = x, u_t = K x_t + L e^i \bigg\},
			\end{aligned}
		\end{equation*}
		where $C^i(\theta)$ is the cost function of the $i$-th set-point tracking problem. By using backward dynamic programming \cite{bertsekas1995dynamic}, it can be shown that $f^i(x;\theta)$ has a quadratic form as $f^i(x;\theta) = x^\top P x + 2 g^{i \top x} + s$, where $P$, $g$ and $s$ are to be determined. By the Bellman equation, for some state $x$ we have
		\begin{equation} \label{eq:bellman_equation}
			\begin{aligned}
				& (x - e^i)^\top Q (x - e^i) + (K x + L e^i)^\top R (K x + L e^i) - C(\theta) \\
				& + \bE [f^i(A x + B (K x + L e^i) + w; \theta)] = f^i(x;\theta).
			\end{aligned}
		\end{equation}
		By substituting the quadratic form of $f^i(x;\theta)$ into \eqref{eq:bellman_equation} and comparing the coefficients of $x^\top x$, $x$ and the constant term, we get the following equations:
		\begin{equation*}
			\begin{aligned}
				P &= Q + K^\top R K + (A + B K)^\top P (A + B K), \\
				g^i &= -Q e^i + K^\top R L e^i + (A + B K)^\top P B L e^i \\
				& ~~~+ (A + B K)^\top g^i, \\
				C^i(\theta) &= e^{i \top} Q e^i + e^{i \top} L^\top R L e^i + e^{i \top} L^\top B^\top P B L e^i \\
				& ~~~ +\mt(P) + 2 g^{i \top} B L e^i.
			\end{aligned}
		\end{equation*}

		Finally, by adding the cost functions $C^i(\theta)$ for all $i=1,\ldots,n$ and using the definitions of $P_K$ and $G_\theta$ in \eqref{eq:policy_matrices}, we obtain the cost funtion $C(\theta)$ in \eqref{eq:compact_lqt}. This completes the proof.

	\section{Proof of Theorem~\ref{thm:equivalence}} \label{app:proof_thm_equivalence}
		Note that
		\begin{equation} \label{eq:data_dynamics_equivalence}
		\begin{aligned}
		&\hA \oX_{0,t} + \hB \oU_{0,t} = [\hB~\hA]\begin{bmatrix} \oU_{0,t} \\ \oX_{0,t} \end{bmatrix} \\
		&~~= (X_1 D_{0,t}^\dagger) (\frac{1}{t} D_{0,t} D_{0,t}^\top) = \oX_{1,t}.
		\end{aligned}
		\end{equation}
		Then, for any $V, H$ and $K, L$ satisfying~\eqref{eq:covariance_parameterization}, it holds that
		\begin{equation} \label{eq:closed_loop_equivalence}
			\begin{aligned}
			\hA + \hB K &= (\hA \oX_{0,t} + \hB \oU_{0,t}) V = \oX_{1,t} V, \\
			\hB L &= (\hA \oX_{0,t} + \hB \oU_{0,t}) H = \oX_{1,t} H.
			\end{aligned}
		\end{equation}
		
		Therefore, by substituting the above equations into \eqref{eq:covariance_parameterization} and \eqref{eq:data_based_lqt}, we can obtain the problem \eqref{eq:compact_lqt}, with $A, B$ substituted by $\hA, \hB$. By Lemma \ref{lem:equivalence} and Lemma \ref{lem:compact_form}, the optimal solution to this substituted problem is $(\hK, \hL)$ with $\hL = \hK_v(I_n - \hA - \hB \hK)^{-\top} Q$. This completes the proof.

		\section{Proof in Section~\ref{sec:offline_deepo}} \label{app:proof_offline_deepo}
		In this appendix, since the model is clear to refer to the estimate model, we omit the hat notation for simplicity, e.g., we use $A, B, C(\theta)$ to denote $\hA, \hB, \hC(\theta)$. 

		\subsection{Proof of Lemma~\ref{lem:policy_gradient}}
		We divide the cost function $C(\xi)$ into three parts as
		\begin{equation*}
			\begin{aligned}
			C_1(\xi) &= \mt(Q + H^\top \oU_0^\top R \oU_0 H + n P_V), \\
			C_2(\xi) &= \mt(2 H^\top \oU_0^\top R \oU_0 V Y_V \oX_1 H - 2 QY_V \oX_1 H), \\
			C_3(\xi) &= \mt(H^\top \oX_1^\top P_V \oX_1 H + 2H^\top \oX_1^\top P_V \oX_1 V Y_V \oX_1 H).
			\end{aligned}
		\end{equation*}
		Then, we compute the gradient of each part separately. For $C_1(\xi)$, we have
		\begin{equation*}
			\nabla_V C_1(\xi) = n \nabla_V \mt(P_V) = 2 n E_V \Sigma_V,
		\end{equation*}
		which follows~\cite[Lemma 2]{zhao2025data}, and
		\begin{equation*}
			\nabla_H C_1(\xi) = 2 \oU_0^\top R \oU_0 H.
		\end{equation*}
		For $C_2(\xi)$, we have
		\begin{equation*}
			\begin{aligned}
			\nabla_V C_2(\xi) &= 2(-\oX_1^\top Y_V^\top Q + \oU_0^\top R\oU_0 H\\
			& + \oX_1^\top Y_V^\top V^\top \oU_0^\top R \oU_0 H ) Z_\xi^\top, \\
			\end{aligned}
		\end{equation*}
		and
		\begin{equation*}
			\begin{aligned}
			\nabla_H C_2(\xi) = &~ 2(-\oX_1^\top Y_V^\top Q + \oX_1^\top Y_V^\top V^\top \oU_0^\top R \oU_0 H \\
			&+ \oU_0^\top R \oU_0 V Y_V \oX_1 H).
			\end{aligned}
		\end{equation*}

		Finally, for $C_3(\xi)$, we have
		\begin{equation*}
			\begin{aligned}
			\nabla_H C_3(\xi) = &~ 2 \big(\oX_1^\top P_V \oX_1 H + \oX_1^\top P_V \oX_1 V Y_V \oX_1 H \\ 
			&~ + \oX_1^\top Y_V^\top V^\top \oX_1^\top P_V \oX_1 H \big),
			\end{aligned}
		\end{equation*}
		By adding the above three parts together, we obtain
		\begin{equation} \label{eq:gradient_H}
			\nabla_H C(\xi) = 2 F_\xi + 2 E_V Z_\xi.
		\end{equation} 

		For $\nabla_V C(\xi)$, it is difficult to culculate the gradient of $C_3(\xi)$ directly and transform it into the form in Lemma~\ref{lem:policy_gradient}. Instead, we use the relationship between $\nabla_V C(\xi)$ and $\nabla_H C(\xi)$ to get the desired form. Note that $Y = (I_n + \oX_1 V Y_V)$, then we have
		\begin{equation*}
			\begin{aligned}
			C_3(\xi) &= \mt(H^\top \oX_1^\top P_V \oX_1 H + 2H^\top \oX_1^\top P_V \oX_1 V Y_V \oX_1 H)\\
			& = \mt(H^\top \oX_1^\top P_V (I_n + \oX_1 V) Y_V \oX_1 H) \\
			& = \mt(H^\top \oX_1^\top Y_V^\top (Q + V^\top \oU_0^\top R \oU_0 V) Y_V \oX_1 H).
			\end{aligned}
		\end{equation*}
		The last equality follows that $P_V = \sum_{i=0}^\infty((\oX_1 V)^\top)^i  (Q + V^\top \oU_0^\top R \oU_0 V) (\oX_1 V)^i$, and $Y_V =  \sum_{i=0}^\infty (\oX_1 V)^i$. By multiplying all the terms of the expansion, and using the fact that $\mt(A^{\top i} \Phi A^j) = \mt(A^{\top j} \Phi A^i)$ for any matrices $A$ and symmetric matrix $\Phi$, we can prove the equation. 

		Next, we compute the gradient of $C_3(\xi)$ with respect to $V$ as
		\begin{equation*}
			\begin{aligned}
			\nabla_V C_3(\xi) = &~ 2 \big(\oX_1^\top Y_V^\top (Q + V^\top \oU_0^\top R \oU_0 V) Y_V \oX_1 H \\
			& + \oU_0^\top R \oU_0 V Y_V \oX_1 H \big)Z_\xi^\top.
			\end{aligned}  
		\end{equation*}
		Meanwhile, we have 
		\begin{equation*}
			\begin{aligned}
			\nabla_H C_3(\xi) = 2 \oX_1^\top Y_V^\top (Q + V^\top \oU_0^\top R \oU_0 V) Y_V \oX_1 H.
			\end{aligned}
		\end{equation*}
		By adding the three parts together, it is clear that $\nabla_V C(\xi) = 2E_V\Sigma_V + \nabla_H C(\xi) Z_\xi^\top$. Then, by substituting \eqref{eq:gradient_H} into the above equation, we obtain the desired form of $\nabla_V C(\xi)$. This completes the proof. 

		\subsection{Proof of Lemma~\ref{lem:equivalence_deepo_model_based}}
		Similar to the proof of Lemma~\ref{lem:policy_gradient}, it can be shown that the policy gradient of the cost function $C(\theta)$ in \eqref{eq:compact_lqt} is given by
		\begin{equation} \label{eq:gradient_KL}
			\nabla_\theta C(\theta) = 2 [E_K ~~ F_\theta] \Phi(\theta),
		\end{equation}
		where $E_K$, $F_\theta$, $\Sigma_K$, $Z_\theta$ are defined similarly as $E_V$, $F_\xi$, $\Sigma_V$, $Z_\xi$, i.e.,
		\begin{equation*}
			\begin{aligned}
			E_K = ~& (R + B^\top P_K B) K + B^\top P_K A, \\
			F_\theta = ~& B^\top G_\theta + R L+ B^\top P_K B L.\\
			\Sigma_K = ~& I_n + (A + B K) \Sigma_K (A + B K)^\top,\\
			Z_\theta = ~& Y_K B L,
			\end{aligned}
		\end{equation*}
		and 
		\begin{equation*}
			\Phi(\theta) = \begin{bmatrix} \Sigma_K + Z_\theta Z_\theta^\top & Z_\theta \\ Z_\theta^\top & I_n \end{bmatrix}.
		\end{equation*}
		By the data-driven policy iteration \eqref{eq:policy_update}, the corresponding updates of $K$ can be written as
		\begin{equation*}
			\begin{aligned}
				&~~~~K^+ - K \\
				&= \oU_0 (V^+ - V)\\
				&= -2\eta \oU_0 (I - \oX_0^\dagger \oX_0) (E_V \Sigma_V + E_V Z_\xi Z_\xi^\top + F_\theta Z_\xi^\top) \\
				&= -2\eta M\left((RK + B^\top P_K (A+B K))(\Sigma_K + Z_\theta Z_\theta^\top)\right.\\
				&~~~\left.+(B^\top F_\theta + R L + B^\top P_K B L)Z_\theta^\top\right),\\
				\end{aligned}
		\end{equation*}
		where the last equality follows from \eqref{eq:data_dynamics_equivalence}, \eqref{eq:closed_loop_equivalence} and \eqref{eq:covariance_parameterization}. Therefore, we have $ K^+ = K - \eta M \nabla_K C(\theta)$.
		
		Similarly, the updates of $L$ can be written as
		\begin{equation*}
			\begin{aligned}
				&~~~~L^+ - L \\
				&= \oU_0 (H^+ - H)\\
				&= -2\eta \oU_0 (I - \oX_0^\dagger \oX_0) ( \oX_1^\top F_\xi + \oU_0^\top R \oU_0 H + \oX_1^\top P_V \oX_1 H) \\
				&= -\eta M \nabla_L C(\theta).
			\end{aligned}
		\end{equation*}

		As for the positive definiteness of $M$. Let $N = \oU_0(I-\oX_0^\dagger \oX_0)$, then we have
		\begin{equation} \label{eq:M_equ_ntn}
		M = N N^\top + \oU_0\oX_0^\dagger\oX_0(I-\oX_0^\dagger \oX_0)\oU_0^\top = N N^\top.
		\end{equation}
		By the persistent excitation condition, we have 
		\[
		\text{rank}(\Lambda) = \text{rank}\begin{bmatrix}
		N + \oU_0\oX_0^\dagger\oX_0 \\ \oX_0
		\end{bmatrix} = \text{rank}\begin{bmatrix}
		N \\ \oX_0
		\end{bmatrix}=n+m.
		\]
		Therefore, we have $\text{rank}(N) = m$. By \eqref{eq:M_equ_ntn}, we know that $M$ is positive definite. This completes the proof.

		\subsection{Proof of Lemma~\ref{lem:gradient_domination_smoothness}}
		We will provide the proof of coercivity later together with the proof of local smoothness. 
		\subsubsection{Gradient Domination}
		In \cite{zhao2023global}, it has been shown that the cost of each subproblem in~\eqref{eq:multi_set_point_lqt} satisfies the following properties (see the proof of Lemma 5 in~\cite{zhao2023global}):
		\begin{equation*}
			C^i(\theta) - C^i(\theta^*) \leq \frac{\Vert \Phi^{i *} \Vert}{\us(R)} \mt \{ [E_K~~F_\theta^i][E_K~~F_\theta^i]^\top\},
		\end{equation*}
		where $\Phi^{i *}$ is the value of 
		\begin{equation*}
			\Phi^i(\theta) = \begin{bmatrix} \Sigma_K + (Z_\theta e^i)(Z_\theta e^i)^\top & Z_\theta e^i \\ (Z_\theta e^i)^\top & 1 \end{bmatrix},
		\end{equation*}
		at the optimal policy $\widehat{\theta}$, and $F_\theta^i = F_\theta e^i$. Moreover, similar to Lemma \ref{lem:policy_gradient}, it can be shown that the policy gradient of $C^i(\theta)$ on $K$ and the $i$-th colume of $L$ is given by
		\begin{equation*}
			\nabla_{[K, L^i]} C^i(\theta) = 2 \begin{bmatrix} E_K & F_\theta^i \end{bmatrix} \Phi^i(\theta).
		\end{equation*}
		Therefore, we can establish the gradient domination property by summing the above inequalities for all $i=1,\ldots,n$ associated
		\begin{equation*}
			C(\theta) - C(\theta^*) \leq \mu(a, \theta) \Vert \nabla_\theta C(\theta) \Vert_F^2,
		\end{equation*}
		where $\mu(a, \theta) = \max_{i} \frac{\Vert \Phi^{i *} \Vert}{4 \us(R) \us(\Phi(\theta))^2}$. By the continuity of the positive definite matrix $\Phi(\theta)$ in $\theta$, it is clear that $\us(\Phi(\theta))$ has a positive lower bound over the compact set $\widehat{\cS}_{\theta} (a)$. Hence, we can choose a uniform constant $\mu(a) > 0$ for all $\theta \in \widehat{\cS}_{\theta} (a)$.

		\subsubsection{Local Smoothness}
		We prove the local smoothness by providing an upperbound for the second-order derivative $\Vert \nabla^2 C(\theta) \Vert := \sup_{\Vert \Delta \Vert_F = 1} \nabla^2 C(\theta)[\Delta, \Delta]$ over the sublevel set $\widehat{\cS}_{\theta} (a)$. In this proof, we omit the dependence of $\theta$ in $E_K$ $F_\theta$, $\Sigma_K$, $Y_K$, $Z_\theta$ and $\Phi(\theta)$ for simplicity.

		First, we have 
		\begin{equation*}
			\nabla C(\theta)[\Delta] = \mt (2\Delta^\top [E ~~ F] \Phi(\theta)).
		\end{equation*}
		Then, we compute $\nabla^2 C(\theta)[\Delta, \Delta]$:
		\begin{equation*}
			\begin{aligned}
				d(E \Sigma &+ E Z Z^\top + F Z^\top) \bigr\vert_K\\
				&= E d\Sigma + dE\bigr\vert_K \Sigma + dE\bigr\vert_K Z Z^\top\\
				&\quad + E d(Z Z^\top)\bigr\vert_K + dF\bigr\vert_K Z^\top + F dZ^\top\bigr\vert_K, \\[4pt]
				d(E Z &+ F) \bigr\vert_K = dE\bigr\vert_K Z + E dZ\bigr\vert_K + dF\bigr\vert_K, \\[4pt]
				d(E \Sigma &+ E Z Z^\top + F Z^\top) \bigr\vert_L\\
				&= E d(Z Z^\top)\bigr\vert_L + dF\bigr\vert_L Z^\top + F dZ^\top\bigr\vert_L,\\[4pt]
				d(E Z &+ F) \bigr\vert_L = E dZ\bigr\vert_L + dF\bigr\vert_L, 
			\end{aligned}
		\end{equation*}
		where
		\begin{equation*}
		\begin{aligned}
			& d\Sigma = \sum_{i=0}^\infty (A+BK)^i \big(BdK \Sigma (A+BK)^\top \\
				& \quad \quad \quad +(A+BK)\Sigma dK^\top B^\top \big) ((A+BK)^\top)^i, \\[4pt]
			&d P = \sum_{i=0}^\infty ((A+BK)^\top)^i \big(dK^\top E + E^\top dK \big) (A+BK)^i, \\[4pt]
			&d E\bigr\vert_K = (R + B^\top P B) dK + B^\top dP (A+BK),\\[4pt]
			&d(Z Z^\top)\bigr\vert_K = Y B dK Z Z^\top + Z Z^\top dK^\top B^\top Y^\top, \\[4pt]
			&d(Z^\top)\bigr\vert_K = Z^\top dK^\top B^\top Y^\top, \\[4pt]
			&dF\bigr\vert_K = B^\top Y^\top dK^\top F + B^\top Y^\top dP B L, \\[4pt]
			&d(Z Z^\top)\bigr\vert_L = Y B dL Z^\top + Z dL^\top B^\top Y^\top, \\[4pt]
			&d(Z^\top)\bigr\vert_L = dL^\top B^\top Y^\top, \\[4pt]
			&dF\bigr\vert_L = B^\top Y^\top P B dL + R dL + B^\top Y^\top K^\top R dL. \\
		\end{aligned}
		\end{equation*}

		In order to continue the derivation, we divide the direction $\Delta$ into two parts as $\Delta = [\Delta_1~~\Delta_2]$, where $\Delta_1$ and $\Delta_2$ are the parts corresponding to $K$ and $L$, respectively. Then, define
		\begin{equation}
		\begin{aligned}
		\gamma_L = & \left( Z^\top \Delta_1^\top E + Z^\top E^\top \Delta_1 + \Delta_2^\top E \right) YB + F^\top \Delta_1 YB \\
				& + (Z^\top \Delta_1 + \Delta_2) \left( B^\top Y^\top PB + R + B^\top Y^\top K^\top R \right) \\[1ex]
		\Sigma_1 = & \sum_{i=0}^\infty (A+BK)^i \bigl[ (A+BK) ( n\Sigma\Delta_1^\top + ZZ^\top \Delta_1^\top \\
				& + Z\Delta_2^\top ) B^\top \bigr] \left( (A+BK)^\top \right)^i \\[1ex]
		\Sigma_2 = & \sum_{i=0}^\infty \left( (A+BK)^\top \right)^i (\Delta_1^\top E) (A+BK)^i \\[1ex]
		\Sigma_3 = & \sum_{i=0}^\infty (A+BK)^i \left( BL(Z^\top \Delta_1^\top + \Delta_2^\top) B^\top Y^\top \right) \\
				& \cdot \left( (A+BK)^\top \right)^i \\[1ex]
		\gamma_K = & \Sigma(A+BK)^\top (\Sigma_2 + \Sigma_2^\top)B + (\Sigma_1 + \Sigma_1^\top)E^\top \\
				& + (\Sigma + ZZ^\top)\Delta_1^\top (R + B^\top PB) \\
				& + YB(\Delta_1 Z + \Delta_2)F^\top + ZF^\top \Delta_1 YB \\
				& + ZZ^\top (E^\top \Delta_1 + \Delta_1^\top E)YB \\
				& + Z\Delta_2^\top (R + B^\top PB + EYB) \\
				& + (\Sigma_3 + \Sigma_3^\top)E^\top.
		\end{aligned}
		\end{equation}
		
		Using the above notations, we can express $\nabla^2 C(\theta)[\Delta, \Delta]$ as
		\begin{equation} \label{eq:direction_derivative}
			\nabla^2 C(\theta)[\Delta, \Delta] = 2 \mt(\Delta_1 \gamma_K + \Delta_2 \gamma_L).
		\end{equation}

		Next, we provide an explicit upper bound for $\Vert \nabla^2 C(\theta) \Vert$. We start by deriving some useful bounds over the sublevel set $\widehat{\cS}_{\theta} (a)$. Notice that each subproblems in \eqref{eq:multi_set_point_lqt} has a stationary state under the policy $\theta$, denoted by $\bar{x}^i = Y B L e^i$. Then, by the definition of $C(\theta)$ in \eqref{eq:compact_lqt}, we can write the cost function in another form as
		\begin{equation} \label{eq:cost_alternative_form}
			\begin{aligned}
			C(\theta) &~= \sum_{i=1}^n \left( (\bar{x}^{i} - e^i)^\top Q (\bar{x}^i-e^i) \right.\\
			&~~~~\left. + (K \bar{x} + L e^i)^\top R (K \bar{x} + L e^i) + tr(P)\right).\\
			&~ = \mt \left( (Z-I_n)^\top Q (Z-I_n) \right. \\
			&~~~~\left. + (KZ + L)^\top R (KZ + L) + n P\right).
			\end{aligned}
		\end{equation}
		Therefore we have $\Vert P \Vert \leq a/n$, $\Vert KZ + L \Vert \leq \sqrt{a / \us(R)}$ and $\Vert Z \Vert \leq \sqrt{a / \us(Q)}$. Also, we have
		\begin{equation*}
			a \geq n \mt(P) \geq n \mt ((Q + K^\top R K) \Sigma) \geq n \us(R) \Vert K \Vert_F^2.
		\end{equation*}
		Then, we have $\Vert \Sigma \Vert \leq a/(n\us(Q))$ and $\Vert K \Vert_F \leq \sqrt{a / (n \us(R))}$. Moreover, we can obtain 
		\begin{equation*}
			\Vert L \Vert \leq \sqrt{a / \us(R)} + \frac{a}{\sqrt{n \us(Q) \us(R)}}.
		\end{equation*}

		Notice that $\Vert Z \Vert \rightarrow \infty$ as $\Vert L \Vert_F \rightarrow \infty$, $\mt(P) \rightarrow \infty$ as $\rho(A+BK) \rightarrow 1$. Therefore, by \eqref{eq:cost_alternative_form}, we can prove the coercivity of $C(\theta)$. Also, since $C(\theta)$ is continuous and the optimal policy $\widehat{\theta}$ is unique, we obtain that the sublevel set $\widehat{\cS}_{\theta} (a)$ is compact.

		Next, we show the upper bound for $Y$. Let $\zeta$ be the eigenvector of $A+BK$ corresponding to the largest eigenvalue. By the definition of $P$ we can get
		\begin{equation*}
			\zeta^\top P \zeta \geq \zeta^\top Q \zeta + \rho(A+BK)^2 \zeta^\top P \zeta,
		\end{equation*}
		which implies that $\rho(A+BK) \leq \sqrt{1 - n\us(Q)/a}$. Moreover, let $\zeta'$ be the right singular vector of $A+BK$ corresponding to the largest singular value, then we have
		\[ 
			\begin{aligned}
			\frac{a}{n} &\geq \zeta'^\top P \zeta' \geq \zeta'^\top Q \zeta' + \Vert A + BK \Vert^2 \us(P) \\
			&\geq (1 + \Vert A + BK \Vert^2) \us(Q),
			\end{aligned}
		\]
		which implies that $\Vert A + BK \Vert \leq \sqrt{a/(n\us(Q)) - 1}$.
		Then, let the schur decomposition of $A + BK$ be $A + BK = U (D + N) U^\top$, where $U$ is an orthogonal matrix, $D$ is a diagonal matrix containing the eigenvalues of $A + BK$, and $N$ is a strictly upper triangular matrix. Then, we have
		\begin{equation} \label{eq:Y_bound}
			\begin{aligned}
			\Vert Y \Vert &~= \Vert (I_n - D - N)^{-1} \Vert \\
			&~\leq \Vert (I_n - D)^{-1}\Vert \sum_{k=0}^{n-1} \Vert N \Vert^k \Vert (I_n - D)^{-1}\Vert^k  \\
			&~\leq \sum_{k=0}^{n-1} \frac{\Vert A+BK \Vert^k}{(1 - \rho(A+BK))^{k+1}}\\
			&~\leq \sum_{k=0}^{n-1} \frac{(a/(n\us(Q)) - 1)^{k/2}}{(1 - \sqrt{1 - n\us(Q)/a})^{k+1}}.
			\end{aligned}
		\end{equation}

		Finally, the upper bound of $\Sigma_{\{1,2,3\}}$ needs to be provided. Given that
		\been
		\begin{aligned}
			&~\Delta_1^\top E + E^\top \Delta_1  \\
			\preceq &~  \Delta_1^\top(R + B^\top P B) \Delta_1 + K^\top R K + (A+BK)^\top P (A+BK)\\
			=  &~ \Delta_1^\top(R + B^\top P B) \Delta_1 + P - Q \\
			\preceq &~ \left(\frac{\Vert R \Vert}{\us(Q)} +\frac{a \Vert B \Vert^2}{n\us(Q)} + \frac{a}{n\us(Q)} -1\right) Q
		\end{aligned}
		\enen
		It follows that:
		\been
		\begin{aligned}
			\Vert \Sigma_2 + \Sigma_2^\top \Vert &~= \sum_{i=0}^\infty ((A+BK)^\top)^i \big( \Delta_1^\top E + E^\top \Delta_1 \big) (A+BK)^i\\
			&~\leq \left\Vert \left(\frac{\Vert R \Vert}{\us(Q)} +\frac{a \Vert B \Vert^2}{n\us(Q)} + \frac{a}{n\us(Q)}-1\right) P \right\Vert\\
			&~\leq \left(\frac{\Vert R \Vert}{\us(Q)} + \frac{a \Vert B \Vert^2}{n\us(Q)} + \frac{a}{n\us(Q)}-1\right) \frac{a}{n}
		\end{aligned}
		\enen
		Furthermore, the relationships between $\Sigma_2$ and both $\Sigma_1$ and $\Sigma_3$ can be established as follows:
		\been
		\begin{aligned}
		&\tr\left((\Sigma_1 + \Sigma_1^\top)E^\top \Delta_1\right) \\
		= & \tr\left((A+BK)(\Sigma \Delta_1^\top + Z Z^\top \Delta_1^\top + Z \Delta_2^\top) B^\top (\Sigma_2 + \Sigma_2^\top)\right)\\[4pt]
		&\tr\left((\Sigma_3 + \Sigma_3^\top)E^\top \Delta_1\right) \\
		= & \tr\left( B L (Z^\top \Delta_1^\top + \Delta_2^\top)B^\top Y^\top (\Sigma_2 + \Sigma_2^\top)\right)
		\end{aligned}
		\enen
		By substituting the previously derived upper bounds into these expressions, we obtain the upper bounds for the components involving $\Sigma_1$ and $\Sigma_3$ in \eqref{eq:direction_derivative}.

		In summary, the upper bound of \eqref{eq:direction_derivative} can be calculated, thereby establishing the Lipschitz smoothness of $C(\theta)$. Due to the extreme complexity of the final expression, its explicit form is not presented here.

		\subsection{Proof of Lemma~\ref{lem:convergence_model_based_po}}
		First, we translate \eqref{eq:model_based_policy_update} into the normal gradient descent form. Let $\psi = M^{-\frac{1}{2}}\theta$. Define $\tilde{C}(\psi) \equiv C(\theta)$ and $\cS_\psi(a) = \{\psi|M^\frac{1}{2}\psi \in \cS_\theta(a)\}$, then 
		\begin{equation} \label{eq:pg_psi}
		\psi^{t+1} = \psi^t - \eta M^{\frac{1}{2}}\nabla_\theta J(\theta^t) = \psi^t - \eta \nabla_\psi \tilde{J}(\psi^t).
		\end{equation}
		By Lemma \ref{lem:gradient_domination_smoothness}, for any $a \geq C^*$, we can show the gradient domination and Lipschitz smoothness property of $\tilde{C}(\psi)$ on $\psi$. We have
		\begin{equation} \label{eq:gd_psi}
		\begin{aligned}
		&\Vert \nabla_\psi \tilde{C}(\psi)\Vert^2 = \Vert M^{\frac{1}{2}} \nabla_\theta C(\theta) \Vert^2 \geq \underline{\sigma}(M)\Vert\nabla_\theta C(\theta) \Vert^2\\
		& ~~~~~~\geq \frac{\underline{\sigma}(M)}{\widehat{\mu}(a)}(C(\theta)-C^*) \geq \frac{\underline{\sigma}(M)}{\widehat{\mu}(a)}(\tilde{C}(\psi)-C^*),
		\end{aligned}
		\end{equation}
		and
		\begin{equation} \label{eq:ls_psi}
		\begin{aligned}
			&\Vert \nabla_\psi \tilde{C}(\psi) - \nabla_\psi \tilde{C}(\psi')\Vert = \Vert M^{\frac{1}{2}} \nabla_\theta C(\theta) - M^{\frac{1}{2}} \nabla_\theta C(\theta')\Vert\\ 
			& ~~~~~~ \leq L(a)\Vert M \Vert^\frac{1}{2} \Vert \theta - \theta'\Vert \leq  L(a)\Vert M \Vert \Vert \psi - \psi'\Vert,
		\end{aligned}
		\end{equation}
		where $\theta$ and $\theta'$ satisfy the requirements in Lemma \ref{lem:gradient_domination_smoothness}. That is, $\tilde{C}(\psi)$ is gradient dominated with constant $\us(M) / \widehat{\mu}(a)$ and smooth with constant $L(a)\Vert M \Vert$ on the set $\cS_\psi(a)$.

		Then, we show that with stepsize $\eta \in \mathopen( 0,1/(L(a) \Vert M \Vert) \mathclose]$, if $\psi^t \in \cS_\psi(a)$ for some $a>J^*$, then $\psi^{t+1}$ is in $\cS_\psi(a)$ by~\eqref{eq:pg_psi}. For simplicity, we use $\psi_{(\eta)}$ to denote $\psi^t - \eta\nabla_\psi \tilde{J}(\psi^t)$.

		By Lemma \ref{lem:gradient_domination_smoothness}, $\cL(a)$ is smooth on $a$. Therefore, given $\phi \in (0,1)$, there exists $b > 0$ such that $L(a+b) \leq (1+\phi) L(a)$. Let $\cS_\psi^c(a+b)$ be the complementary set of $\cS_\psi(a+b)$. Clearly, the distance $d = \text{inf}\left\{\Vert\psi - \psi'\Vert\big|\forall \psi \in \cS_\psi(a) \text{ and } \psi' \in \cS_\psi^c(a+b)\right\}$ is larger than $0$. Choose a large enough $\overline{N}$ such that $2/\left(\overline{N}(1+\phi)L(a)\Vert M \Vert\right) < d/\Vert \nabla_\psi \tilde{J}(\psi^t) \Vert$, where $\Vert \nabla_\psi \tilde{J}(\psi^t) \Vert > 0$ is ensured by \eqref{eq:gd_psi} and $a>J^*$, and choose $\tau \in \left(0,2/(\overline{N}(1+\phi)L(a)\Vert M \Vert)\right]$. Then, it holds that $\Vert \psi_{(\tau)} - \psi^t \Vert < d$, which implies $\psi_{(\tau)} \in \cS_\psi(a+b)$. By the smoothness \eqref{eq:ls_psi} on $\cS_\psi(a+b)$, we have $\tilde{J}(\psi_{(\tau)}) - \tilde{J}(\psi^t) \leq -\tau(1-\tau L(a+b)\Vert M \Vert/2)\Vert \nabla_\psi \tilde{J}(\psi^t) \Vert^2\leq0$, which implies that $\psi_{(\tau)} \in \cS_\psi(a)$. Similarly, it holds that $\Vert \psi_{(2\tau)} - \psi_{(\tau)} \Vert < d$, which implies $\psi_{(2\tau)} \in \cS_\psi(a+b)$. By induction we can show that $\psi_{(N\tau)} \in \cS_\psi(a)$ for some $N\in \bN_+$ as long as $\tau N L(a+b)\Vert M \Vert/2 < 1$. Since $L(a+b) \leq (1+\phi)L(a)<2L(a)$, for $\eta$ in $(0,1/(L(a) \Vert M \Vert)]$, we can choose $\tau$, $N$ such that $\eta = N\tau$ to show $\psi^{t+1} = \psi_{(\eta)} \in \cS_\psi(a)$.

		Next, we show that \eqref{eq:pg_psi} converges linearly. With the stepsize $\eta \in (0,1/(L^0 \Vert M \Vert)]$, the update \eqref{eq:pg_psi} satisfies $\psi^t \in \cS_\psi(J(\theta^0)), \forall t\in\bN$. And the cost satisfies 
		\begin{equation*}
		\tilde{C}(\psi^{t+1}) \leq \tilde{C}(\psi^t) - \eta\left(1-\frac{L^0\Vert M \Vert \eta}{2}\right)\Vert\nabla_\psi \tilde{C}(\psi^t)\Vert^2.
		\end{equation*}
		Using the gradient domination property \eqref{eq:gd_psi}, and reorganize the terms we obtain 
		\begin{equation}
		\tilde{C}(\psi^{t+1}) - J^* \leq \left(1 -\nu\right)(\tilde{C}(\psi^t)-J^*).
		\end{equation}
		Since $\tilde{C}(\psi) \equiv C(\theta)$, We finally obtain
		\begin{equation*}
		C(\theta^{t+1}) - C^* \leq \left(1 -\nu\right)(C(\theta^t)-C^*).
		\end{equation*}

	\section{Proof of Theorem \ref{thm:convergence_online_deepo}}

	The proof of Theorem \ref{thm:convergence_online_deepo} is as follows. First, we establish the equivalence between the data-driven projected gradient descent and the model-based gradient descent. Let $\hC_{t}(\theta)$ denote the estimated value of the cost function $C(\theta)$ under the estimated model $\hA_t, \hB_t$ given the data at time $t$. We then obtain the following lemma.

	\begin{lemma} \label{lem:equivalence_deepo_model_based_2}
		The policy iteration in Algorithm \ref{alg:adaptive_deepo_dynamic} is equivalent to the following model-based policy gradient update:
		\begin{equation} \label{eq:model_based_policy_update_2}
			\theta_{t+1} = \theta_t - \eta M_{t+1} \nabla_\theta \hC_{t+1}(\theta_t),
		\end{equation}
		and $M_t$ satisfies $\us(M_t)\geq \gamma_t^4$.
	\end{lemma}

	\begin{pf}
		The first half of the theorem is identical to the proof of Lemma \ref{lem:equivalence_deepo_model_based} and is omitted here. Regarding the minimum singular value of $M_t$, note that
		$$
			\Lambda_t \Pi_{\oX_{0,t+1}} \Lambda_t^\top = \begin{bmatrix}
				M_t & \boldsymbol{0} \\
				\boldsymbol{0} & \boldsymbol{0}
			\end{bmatrix}.
		$$
		Hence, we have $\us(M_t) = \sigma_m\left(\Lambda_t \Pi_{\oX_{0,t+1}} \Lambda_t^\top\right)$. According to the properties of projection matrices, it is known that $\sigma_m\left(\Pi_{\oX_{0,t+1}}\right) = 1$, which yields
		$$
			\us(M_t) \geq \us(\Lambda_t)\sigma_m\left(\Pi_{\oX_{0,t+1}}\right)\us(\Lambda_t) \geq \gamma_t^4.
		$$
	\end{pf}

	Based on Lemma \ref{lem:equivalence_deepo_model_based_2}, it suffices to consider the model-based policy update. Similar to the properties of $\hC(\theta)$ established in Lemma \ref{lem:gradient_domination_smoothness}, the cost function $C(\theta)$ under the true system is also gradient dominated and Lipschitz continuous. Let $C^*$ denote the optimal cost of \eqref{eq:compact_lqt}. The following lemma holds.

	\begin{lemma}\label{lem:gradient_domination_smoothness_2}
		For any $a > C^*$ and any $\theta \in \cS_{\theta} (a)$, the following properties hold:
		\begin{itemize}
			\item (Gradient Domination) There exists a constant $\mu(a) > 0$ such that
			\begin{equation*}
				C(\theta) - C^* \leq \mu(a) \Vert \nabla_\theta C(\theta) \Vert_F^2,
			\end{equation*}
			where $\widehat{\theta}^*$ is the optimal policy for \eqref{eq:compact_lqt}.
			\item (Local Lipschitz Smoothness) For any $\theta', \theta \in \cS_{\theta} (a)$ satisfying $$\theta + b(\theta' - \theta) \in \cS_{\theta} (a), \forall b \in [0,1],$$ there exists a constant $L(a) > 0$, polynomial in $a$, such that
			\begin{equation*}
				C(\theta') \leq C(\theta) + \langle \nabla_\theta C(\theta), \theta' - \theta \rangle + \frac{L(a)}{2} \Vert \theta' - \theta \Vert_F^2.
			\end{equation*}
		\end{itemize}
	\end{lemma}

	Next, we consider a single step update at an arbitrary time. Suppose the policy prior to the update is $\theta$, and take an arbitrary positive definite scaling matrix $M$, assuming the constants in Assumptions~\ref{asm:noise} and~\ref{asm:input} are $\delta$ and $\gamma$, respectively. Then, the update based on the estimated model $(\hA,\hB)$ is
	\begin{equation} \label{eq:update_estimate_system}
		\theta' = \theta - \eta M \nabla_\theta \hC(\theta).
	\end{equation}
	Meanwhile, based on the true system $(A,B)$, another policy update formula is obtained:
	\begin{equation} \label{eq:update_true_system}
		\theta'' = \theta - \eta M \nabla_\theta C(\theta).
	\end{equation}
	The following lemma establishes the convergence result for the policy iteration in \eqref{eq:update_true_system}.

	\begin{lemma} \label{lem:convergence_true_gradient}
		For a policy $\theta \in \cS$ and a step size $\eta \in \mathopen(0, 1/(L(C(\theta)) \Vert M \Vert) \mathclose]$, the policy update \eqref{eq:update_true_system} satisfies
		\[	C(\theta'') - C^* \leq \left(1 - \frac{\eta\us(M)}{2\mu(C(\theta))}\right)(C(\theta) - C^*).
		\]
	\end{lemma}
	\begin{pf}
		Similar to the proof of Lemma \ref{lem:convergence_model_based_po}, we have
		\[
		\begin{aligned}
			&C(\theta'') - C^* \\
			\leq~& \left(1-\eta \left(1 - \frac{L(C(\theta))\eta\Vert M \Vert}{2}\right)\frac{\us(M)}{\mu(C(\theta))}\right)(C(\theta) - C^*)\\
			\leq~& \left(1 - \frac{\eta\us(M)}{2\mu(C(\theta))}\right)(C(\theta) - C^*).
		\end{aligned}
		\]
	\end{pf}

	Algorithm \ref{alg:adaptive_deepo_dynamic} does not utilize the true system $(A,B)$ to compute the policy gradient $\nabla_\theta C(\theta)$ for policy updates; instead, it follows \eqref{eq:update_estimate_system}. Therefore, it is necessary to prove the convergence of $\theta'$. We will demonstrate that when the SNR is sufficiently small, the difference between $\theta'$ and $\theta''$ is bounded, which will be achieved through the following series of lemmas. First, Lemma \ref{lem:estimate_error} indicates that the gap between the estimated system and the true system is directly bounded by the inverse SNR.

	\begin{lemma} \label{lem:estimate_error}
		The system estimate $(\hA,\hB)$ satisfies
		\[\Vert [\hB, \hA] - [B,A] \Vert \leq \frac{\delta}{\gamma}.
		\]
	\end{lemma}
	\begin{pf}
		This follows trivially from Assumption~\ref{asm:noise} and Assumption~\ref{asm:input}.
	\end{pf}

	Second, the following lemma establishes the upper bounds for the variations of $\Sigma$ and $Y$ under small policy perturbations.
	\begin{lemma} \label{lem:lyapunov_perturbation}
		Let $A$ be a stable matrix, i.e., $\rho(A)<1$, and define $\Sigma(A) = I + A \Sigma(A) A^\top$, $Y(A) = (I - A)^{-1}$. If $$\Vert A' - A \Vert \leq \text{min}\left\{\frac{1}{(4\Vert \Sigma(A) \Vert(1+\Vert A \Vert))}, \frac{1}{2\Vert Y(A) \Vert}\right\},$$ then $A'$ is also stable, and it holds that $\Vert \Sigma(A') - \Sigma(A) \Vert \leq 4\Vert \Sigma(A) \Vert^2(1+\Vert A \Vert)\Vert A' - A \Vert$, as well as $\Vert Y(A) - Y(A') \Vert \leq 2 \Vert Y(A) \Vert^2 \Vert A - A'\Vert$.
	\end{lemma}
	\begin{pf}
		The part concerning $\Sigma$ is consistent with \cite[Lemma 15]{zhao2025data}. We focus here on the part concerning $Y$. Since
		$$ \begin{aligned}
		&\Vert (I - Y(A)(A'-A))^{-1} \Vert \\
		\leq~& \Vert I + Y(A)(A'-A)(I - Y(A)(A'-A))^{-1} \Vert \\
		\leq~& 1 + \Vert Y(A) (A'-A) \Vert \Vert (I - Y(A)(A'-A))^{-1} \Vert \\
		\leq~& 1 + \frac{1}{2}\Vert (I - Y(A)(A'-A))^{-1} \Vert,
		\end{aligned}$$
		it implies $\Vert (I - Y(A)(A'-A))^{-1} \Vert \leq 2$. Consequently, we have
		$$
		\begin{aligned}
		&\Vert Y(A) - Y(A')\Vert \\
		=~& \Vert \left(I - (I - Y(A)(A'-A))^{-1}\right)Y(A)\Vert\\
		\leq~& \Vert \left(I - (I - Y(A)(A'-A))^{-1}\right)\Vert \Vert Y(A)\Vert \\
		\leq~& \Vert Y(A) \Vert^2 \Vert A'-A \Vert \Vert (I - Y(A)(A'-A))^{-1} \Vert\\
		\leq~& 2 \Vert Y(A) \Vert^2 \Vert A - A'\Vert.
		\end{aligned}
		$$
	\end{pf}

	Third, we prove that the difference between the policy gradients $\nabla_\theta C(\theta)$ and $\nabla_\theta \hC(\theta)$ in \eqref{eq:update_estimate_system} and \eqref{eq:update_true_system} is bounded. We define the following parameter:
	\[ \begin{aligned}
	p_1 = &\frac{\sqrt{n\us(R)}}{\sqrt{n\us(R)} + \sqrt{C(\theta)}} \cdot \text{min}\left\{\frac{n \us(Q)}{4C(\theta)\left(1 + \frac{C(\theta)}{n\us(Q)}\right)}, \right. \\
	&\left. \sum_{k=0}^{n-1} \frac{\left(1 - \sqrt{1 - n\us(Q)/C(\theta)}\right)^{k+1}}{2\left(C(\theta)/(n\us(Q)) - 1\right)^{k/2}} \right\}. 
	\end{aligned}
	\]
	We then arrive at the following lemma.

	\begin{lemma} \label{lem:gradient_error}
		Let $\theta\in\cS_\theta$. There exists a polynomial $p_2 = \text{poly}(C(\theta), \Vert A \Vert, \Vert B \Vert, \Vert R \Vert,$ $ 1/\us(Q), 1/\us(R))$ such that: if $\delta/\gamma \leq p_1$, then $\Vert \nabla_\theta C(\theta) - \nabla_\theta \hC(\theta) \Vert \leq p_2 \delta/\gamma$.
	\end{lemma}

	\begin{pf}
		For brevity, the subscripts $\theta$ and $K$ denoting policy dependence on matrices are omitted in this proof. According to the gradient expression \eqref{eq:gradient_KL}, we have
		\[
			\begin{aligned}
			&\nabla_K C(\theta) - \nabla_K \hC(\theta) = \\
			&~~~~~2\left(E \Sigma + E Z Z^\top + F Z^\top - \hE \widehat{\Sigma} - \hE \hZ \hZ^\top - \hF \hZ^\top\right)\\
			&\nabla_L C(\theta) - \nabla_L \hC(\theta) = 2\left(
			E Z+ F - \hE \hZ- \hF\right)
			\end{aligned}
		\]
		We next prove that each of the terms $\Vert E \Sigma - \hE \widehat{\Sigma}\Vert$, $\Vert E Z Z^\top- \hE \hZ \hZ^\top\Vert$, $\Vert F Z^\top - \hF \hZ^\top\Vert$, $\Vert EZ - \hE \hZ\Vert$, and $\Vert F - \hF\Vert$ is bounded by a polynomial multiple of the inverse SNR.
		
		According to Lemma \ref{lem:estimate_error}, we have
		\begin{align}
			&\Vert (A+BK) - (\hA + \hB K) \Vert \leq (1 + \Vert K \Vert)\delta/\gamma, \label{eq:K_error}\\
			&\Vert BL - \hB L \Vert \leq \Vert L \Vert\delta/\gamma. \label{eq:L_error}
		\end{align}
		
		The term $\Vert E \Sigma - \hE \widehat{\Sigma}\Vert$ coincides with the difference in the cost function for the LQR problem, the proof of which can be found in~\cite[Lemma 12]{zhao2025policy} and is therefore omitted here. In the proof, it has also been shown that both $\Vert E - \hE \Vert$ and $\Vert P - \hP \Vert$ admit polynomial upper bounds.
		
		For $\Vert EZ - \hE \hZ\Vert$, we have \[EZ - \hE \hZ = E(Z-\hZ) + (E-\hE)\hZ.\] By \cite[Lemma 11]{fazel2018global}, it holds that 
		\[\tr(E^\top E) \leq \Vert R + B^\top P B \Vert (C(K) - C^*).\]
		Consequently, $ \Vert E \Vert \leq \Vert R + B^\top P V \Vert^{\frac{1}{2}}(C(K) - C^*)^{\frac{1}{2}}$. Additionally, we have
		\[ Z - \hZ = Y(BL - \hB L) + (Y - \hY)\hB L,\]
		where, according to Lemma \ref{lem:lyapunov_perturbation}, we obtain 
		\[
		\begin{aligned}
		\Vert Y - \hY \Vert &~\leq 2 \Vert Y \Vert^2 (1 + \Vert K \Vert)\delta/\gamma \\
		&~\leq \Vert Y \Vert^2 \left(1 + \sqrt{\frac{C(\theta)}{(n\us(R))}}\right)p_1.
		\end{aligned}
		\]
		Combining \eqref{eq:L_error}, the upper bound of $\Vert Y \Vert$ from \eqref{eq:Y_bound}, and the upper bound of $\Vert L \Vert$, a polynomial upper bound can be established for $Z - \hZ$, and thus for $\Vert EZ - \hE \hZ\Vert$.
		
		For $\Vert F - \hF \Vert$, we have
		$$	F - \hF = (B^\top G - \hB^\top \hG) + (B^\top P B - \hB^\top \hP \hB) L. $$ 
		Similar to the preceding steps, because each component in $F - \hF$ possesses a polynomial bound, proving that $\Vert F - \hF \Vert$ has a polynomial upper bound requires to prove that $\Vert B - \hB\Vert$, $\Vert G - \hG\Vert$, and $\Vert P - \hP \Vert$ have polynomial bounds. By the same logic, finding a polynomial bound for $\Vert G - \hG\Vert$ requires finding polynomial bounds for $\Vert Y - \hY\Vert$, $\Vert P - \hP \Vert$, $\Vert (A+BK) - (\hA + \hB K) \Vert$, and $\Vert BL - \hB L\Vert$. Since their bounds have already been established, they are not explicitly expanded here for brevity.
		
		Finally, based on $E Z Z^\top- \hE \hZ \hZ^\top = E Z (Z^\top - \hZ^\top) + (E Z - \hE \hZ) \hZ^\top$ and $F Z^\top - \hF \hZ^\top = F(Z - \hZ)^\top + (F - \hF) \hZ^\top$, the upper bounds for these two terms can be similarly synthesized from the bounds derived above.
	\end{pf}

	Fourth, we prove that when two policies $\theta$ and $\tilde{\theta}$ are sufficiently close, the gap between their cost functions is also bounded.
	\begin{lemma} \label{lem:cost_error}
		Let $\theta \in \cS_\theta$. There exist polynomials $p_3 = \text{poly}(C(\theta), 1/\Vert A \Vert, 1/\Vert B \Vert, 1/\Vert R \Vert,$ $ 1/\us(Q), \us(R))$ and $p_4, p_5 =  \text{poly}(C(\theta), \Vert A \Vert, \Vert B \Vert$ $, \Vert R \Vert, 1/\us(Q), 1/\us(R))$ such that: if $\Vert \tilde{\theta} - \theta \Vert \leq p_3$, then $\tilde{\theta}\in\cS$ and
		$$ \Vert \tilde{\Sigma} - \Sigma \Vert \leq p_4, \vert C(\tilde{\theta}) - C(\theta) \vert \leq p_5 \Vert \tilde{\theta} - \theta \Vert.$$
	\end{lemma}

	\begin{pf}
		According to Lemma \ref{lem:lyapunov_perturbation}, there exist $p_3, p_4, p$ such that $\tilde{\theta}\in\cS$ and $ \Vert \tilde{\Sigma} - \Sigma \Vert \leq p_4$, $\Vert \tilde{Y} - Y \Vert \leq p$. Furthermore,
		\begin{equation}
		\begin{aligned}
		C(\tilde{\theta}) - C(\theta) = & \operatorname{tr} \Big( 
		\tilde{L}^\top R \tilde{L} - L^\top R L \\
		& + \tilde{L}^\top B^\top P_{\tilde{K}} B \tilde{L} - L^\top B^\top P_K B L \\
		& + n(P_{\tilde{K}} - P_K) + 2(G_{\tilde{\theta}}^\top B \tilde{L} - G_\theta^\top B L) \Big).
		\end{aligned}
		\end{equation}
		Since each term in the above equation has a polynomial upper bound, employing proof techniques similar to those in Lemma \ref{lem:gradient_error}, it suffices to find polynomial upper bounds for $\Vert \tilde{L} - L \Vert$, $\Vert P_{\tilde{K}} - P_K \Vert$, and $\Vert G_{\tilde{\theta}} - G_\theta \Vert$. The bound for $\Vert \tilde{L} - L \Vert$ can be derived using $p_3$, $\Vert P_{\tilde{K}} - P_K \Vert$ using $p_4$, and $\Vert G_{\tilde{\theta}} - G_\theta \Vert$ using $p_3, p_4, p$. Due to their complex explicit forms, they are omitted here.
	\end{pf}

	Fifth, we prove that when $\delta/\gamma$ is sufficiently small, the difference between $C(\theta'')$ and $C(\theta')$ obtained from \eqref{eq:update_estimate_system} and \eqref{eq:update_true_system} is bounded.

	\begin{lemma} \label{lem:cost_error_2}
		Let $\theta \in \cS$. There exists a polynomial $p_6 = \text{poly}(C(\theta), 1/\Vert A \Vert, 1/\Vert B \Vert, 1/\Vert R \Vert,$ $ 1/\us(Q), \us(R))$ such that: if
		$$ \delta/\gamma \leq p_1 \text{ and } \eta \leq \frac{1}{\Vert M \Vert} \cdot \text{min}\left\{\frac{p_3 \gamma}{p_2\delta}, p_6\right\},$$
		then $\vert C(\theta'') - C(\theta') \vert \leq \eta p_2 p_5 \Vert M \Vert \delta/\gamma$.
	\end{lemma}

	\begin{pf}
		By Lemma \ref{lem:gradient_error}, we have
		$$ \Vert \theta' - \theta'' \Vert = \eta \Vert M\nabla C(\theta) - M\nabla \hC(\theta) \Vert \leq \frac{\eta \Vert M \Vert p_2 \delta}{\gamma} \leq p_3.$$
		Let $p_6 = 1/L(C(\theta))$. According to Lemma \ref{lem:convergence_true_gradient}, it holds $\theta'' \in \cS$. Together with Lemma \ref{lem:cost_error}, the proof is completed.
	\end{pf}

	Through Lemma \ref{lem:convergence_true_gradient} and Lemma \ref{lem:cost_error_2}, we can formulate the convergence result for \eqref{eq:update_estimate_system}:
	\begin{equation} \label{eq:convergence_C_theta_prime}
	C(\theta') - C(\theta) \leq -\frac{\eta \us(M)}{2\mu(C(\theta))}(C(\theta) - C^*) + \frac{\eta p_2 p_5 \Vert M \Vert \delta}{\gamma}.
	\end{equation}
	However, to obtain the convergence result as stated in Theorem \ref{thm:convergence_online_deepo}, it remains necessary to find a uniform upper bound for the cost function $C(\theta_t)$ throughout the iteration process. This will provide bounding limits (upper or lower) for the parameters $p_i, i = 1,\dots,6$, as well as $L(C(\theta))$ and $\mu(C(\theta))$. Let
	\[ \overline{C} = C^* + C(\theta_{t_0}) + 1 + \frac{1}{2\underline{L}\underline{\mu}},\]
	where $\underline{L} := L(C^*)$ and $\underline{\mu} := \mu(C^*)$. In addition, define $\overline{L} = L(\overline{C})$ and $\overline{\mu} = \mu(\overline{C})$, and let $\underline{p}_1$, $\overline{p}_2$, $\overline{p}_4$, $\overline{p}_5$, and $\underline{p}_6$ represent the values of these quantities evaluated at $\overline{C}$, while $\underline{p}_3$ is evaluated at $C^*$.

	\begin{lemma}
		If
		$$\frac{\delta_t}{\gamma_t} \leq \text{min}\left\{\underline{p}_1, \frac{\us(M_t)}{2\overline{p}_5\overline{p}_2\overline{\mu}\Vert M_t \Vert}\right\} $$
		and 
		$$ \eta \leq \frac{1}{\Vert M_t\Vert} \cdot \text{min}\left\{ \frac{\underline{p}_3 \gamma_t}{\overline{p}_2 \delta_t}, \underline{p}_6 \right\},$$
		then for all time steps $t\geq t_0$, it holds that $C(\theta_t) \leq \overline{C}$.
	\end{lemma}

	\begin{pf}
		First, at time $t_0$, $C(\theta_{t_0}) \leq \overline{C}$ trivially holds. Assume that at time $t \geq t_0$, $C(\theta_t) \leq \overline{C}$ still holds. By \eqref{eq:convergence_C_theta_prime}, we have
		$$
		\begin{aligned}
			&C(\theta_{t+1}) - C(\theta_t) \\
			\leq & -\frac{\eta \us(M_{t+1})}{2\overline{\mu}}(C(\theta_t) - C^*) + \eta \overline{p}_2 \overline{p}_5 \Vert M_{t+1} \Vert \frac{\delta_{t+1}}{\gamma_{t+1}}\\
			\leq & -\frac{\eta \us(M_{t+1})}{2\overline{\mu}}(C(\theta_t) - C^*) + \frac{\eta \us(M_{t+1})}{2 \overline{\mu}}.
		\end{aligned}
		$$
		We proceed by evaluating cases. If $C(\theta_t) \leq C^* + 1$, then
		$$ C(\theta_{t+1}) \leq C(\theta_t) - \frac{\eta \us(M_{t+1})}{2 \overline{\mu}} + \frac{\eta \us(M_{t+1})}{2 \overline{\mu}} = C(\theta_t) \leq \overline{C}.$$
		Conversely, if $C(\theta_t) \geq C^* + 1$, then
		\[ \begin{aligned}
		C(\theta_{t+1}) &\leq C(\theta_t) + 1 + \frac{\eta \us(M_{t+1})}{2 \overline{\mu}} \\
		&\leq C^* + 1 + \frac{\underline{p}_6}{2 \overline{\mu}} = C^* + 1 + \frac{1}{2 \overline{L} \overline{\mu}} \leq \overline{C}.
		\end{aligned}
		\]
		Therefore, by mathematical induction, the lemma is proven.
	\end{pf}

	In conclusion, the convergence result \eqref{eq:convergence_C_theta_prime} holds consistently throughout the iteration process. By consolidating and rearranging the iteration formulas from all past time steps, the convergence result depicted in Theorem \ref{thm:convergence_online_deepo} is obtained.
		
	\bibliographystyle{agsm}
	\bibliography{mybib}
	
\end{document}